\documentclass[11pt,11pt,letterpaper,reqno]{amsart}
\usepackage[utf8]{inputenc}
\usepackage{xcolor}
\usepackage{mathrsfs}
\usepackage{enumitem}
\usepackage{amstext}
\usepackage{amsthm}
\usepackage{amssymb}
\usepackage{graphicx}
\usepackage{geometry}
\PassOptionsToPackage{normalem}{ulem}
\usepackage{ulem}
\usepackage[bookmarks=false,
 breaklinks=false,pdfborder={0 0 1},backref=false,colorlinks=false]
 {hyperref}

\makeatletter

\providecolor{lyxadded}{rgb}{0,0,1}
\providecolor{lyxdeleted}{rgb}{1,0,0}
\DeclareRobustCommand{\mklyxadded}[1]{\textcolor{lyxadded}\bgroup#1\egroup}
\DeclareRobustCommand{\mklyxdeleted}[1]{\textcolor{lyxdeleted}\bgroup\mklyxsout{#1}\egroup}
\DeclareRobustCommand{\mklyxsout}[1]{\ifx\\#1\else\sout{#1}\fi}

\numberwithin{equation}{section}
\numberwithin{figure}{section}

\usepackage{amscd}
\usepackage{latexsym}
\usepackage{amscd}
\usepackage[mathscr]{euscript}
\usepackage{mathrsfs}
\usepackage{xcolor}
\usepackage{color}

\usepackage{lscape}
\usepackage{epstopdf}
\usepackage{array}
\usepackage{tabularx}
\usepackage{longtable}
\usepackage{ltxtable}
\usepackage{color}

\usepackage{tikz}
\usetikzlibrary{decorations,arrows}
\usetikzlibrary{decorations.pathmorphing}
\usepgflibrary{decorations.pathreplacing}
\usepackage{pgf}
\usetikzlibrary{patterns}
\tikzset{
  schraffiert/.style={pattern=horizontal lines,pattern color=#1},
  schraffiert/.default=black
}
\tikzset{
    ultra thin/.style= {line width=0.1pt},
    very thin/.style=  {line width=0.2pt},
    thin/.style=       {line width=0.4pt},
    semithick/.style=  {line width=0.6pt},
    thick/.style=      {line width=0.8pt},
    very thick/.style= {line width=1.2pt},
    ultra thick/.style={line width=2.4pt}
}
 \usetikzlibrary{shapes}
\definecolor{hellgrau}{rgb}{0.93,0.93,0.93}
\definecolor{hellergrau}{rgb}{0.97,0.97,0.97}
\definecolor{hellgruen}{rgb}{0.6,1.35,0.5}
\definecolor{grau}{rgb}{0.93,0.93,0.93}
\definecolor{hellblau}{rgb}{0.8,0.8,2.0}
\definecolor{blau}{rgb}{0.3,0.5,2.0}
\definecolor{hellrot}{rgb}{2.0,0.6,0.6}
\definecolor{gruen}{rgb}{0.3,0.75,0.2}
\definecolor{rot}{rgb}{0.9,0.1,0.1}
\definecolor{orang}{rgb}{1.3,0.65,0}
\allowdisplaybreaks
\DeclareMathAlphabet{\mathpzc}{OT1}{pzc}{m}{it}
\newcommand{\Aa}{{\mathcal A}}

\newcommand{\precd}{\prec\!\!\prec}					
\newcommand{\strict}{\subseteq_{\text{strict}}} 	
\newcommand{\ind}{\mbox{ind}}						
\newcommand{\indB}{\mbox{ind}_B}						
\newcommand{\motw}{=_2}

\allowdisplaybreaks
\usepackage{xcolor}
\newcommand{\Hmm}[1]{\leavevmode{\marginpar{\tiny%
$\hbox to 0mm{\hspace*{-0.5mm}$\leftarrow$\hss}%
\vcenter{\vrule depth 0.1mm height 0.1mm width \the\marginparwidth}%
\hbox to 0mm{\hss$\rightarrow$\hspace*{-0.5mm}}$\\\relax\raggedright #1}}}
\definecolor{green}{RGB}{0, 180, 0}
\definecolor{cyan}{RGB}{0, 180, 180}
\definecolor{yellow}{RGB}{211,211,0}

\makeatother

\theoremstyle{plain}
\newtheorem{thm}{\protect\theoremname}[section]
\theoremstyle{remark}
\newtheorem*{rem*}{\protect\remarkname}
\theoremstyle{plain}
\newtheorem{prop}[thm]{\protect\propositionname}
\theoremstyle{definition}
\newtheorem{defn}[thm]{\protect\definitionname}
\newtheorem{example}[thm]{\protect\examplename}
\theoremstyle{plain}
\newtheorem{lem}[thm]{\protect\lemmaname}
\theoremstyle{remark}
\newtheorem{rem}[thm]{\protect\remarkname}
\providecommand{\definitionname}{Definition}
\providecommand{\examplename}{Example}
\providecommand{\lemmaname}{Lemma}
\providecommand{\propositionname}{Proposition}
\providecommand{\remarkname}{Remark}
\providecommand{\theoremname}{Theorem}

\begin{document}

\global\long\def\theenumi{\alph{enumi}}%

\global\long\def\ui{\mathbf{\textrm{i}}}%

\global\long\def\ue{\mathbf{\textrm{e}}}%

\global\long\def\ud{\mathbf{\textrm{d}}}%

\global\long\def\sgn{\mathrm{sign}}%

\global\long\def\id{\mathbf{1}}%

\global\long\def\C{\mathbb{C}}%

\global\long\def\R{\mathbb{R}}%

\global\long\def\Q{\mathbb{Q}}%

\global\long\def\N{\mathbb{N}}%

\global\long\def\Z{\mathbb{Z}}%

\global\long\def\Id{\mathbf{\mathbb{I}}}%

\global\long\def\V{\mathcal{V}}%


\global\long\def\Aa{\mathcal{A}}%

\global\long\def\tree{\mathcal{T}_{\alpha}}%

\global\long\def\graph{\mathcal{G}_{\alpha}}%

\global\long\def\paths{\Gamma_{\alpha}}%

\global\long\def\emap{E_{\alpha}}%

\global\long\def\Cen{Z}%

\global\long\def\co{\textbf{c}}%

\global\long\def\cop{\textbf{\ensuremath{\widetilde{\co}}}}%

\global\long\def\cm{[\textbf{c},m]}%

\global\long\def\cmn{[\textbf{c},m,n]}%

\global\long\def\Co{\mathscr{C}}%

\global\long\def\emp{\emptyset}%

\global\long\def\tr{\textit{tr}}%

\global\long\def\Nz{\N_{0}}%

\global\long\def\Nmo{\N_{-1}}%


\global\long\def\precd{\prec_{\mathrm{str}}}%

\global\long\def\strict{\subseteq_{\mathrm{str}}}%

\global\long\def\Mult{\textrm{\ensuremath{\mathcal{M}}}}%

\global\long\def\relA{\mbox{ind}^{A}_{\textrm{rel}}}%

\global\long\def\ind{\mbox{ind}}%

\global\long\def\indA{\mbox{ind}_{A}}%

\global\long\def\indB{\mbox{ind}_{B}}%

\global\long\def\indC{\mbox{ind}_{C}}%

\global\long\def\dR{\delta_{R}}%

\global\long\def\motw{=_{2}}%

\global\long\def\spec{\mathrm{spec}}%

\global\long\def\tA{\mathit{A}}%

\global\long\def\tB{\mathit{B}}%

\global\long\def\tC{\mathit{G}}%

\global\long\def\sigc{\sigma_{\co}}%

\global\long\def\sigcm{\sigma_{[\co,m]}}%

\global\long\def\sigcmn{\sigma_{[\co,m,n]}}%

\global\long\def\sigk{\sigma_{k}}%

\global\long\def\sigkm{\sigma_{k-1}}%

\global\long\def\sigkmm{\sigma_{k-2}}%

\global\long\def\sigkp{\sigma_{k+1}}%

\global\long\def\sigkpp{\sigma_{k+2}}%

\global\long\def\alk{\alpha_{k}}%

\global\long\def\Ham{H_{\alpha,V}}%

\global\long\def\Hrat{H_{\frac{p}{q},V}}%

\global\long\def\Halk{H_{\alk,V}}%

\global\long\def\nHcV{H^{\times n}_{\co,V}}%

\global\long\def\nHc{H^{\times n}_{\co}}%

\global\long\def\nHcm{H^{\times n}_{[\co,m]}}%

\global\long\def\nHcmV{H^{\times n}_{[\co,m],V}}%

\global\long\def\Sturm{\omega_{\alpha}}%

\global\long\def\IDS{N_{\alpha,V}}%

\global\long\def\Hc{H_{\co}}%

\global\long\def\HcV{H_{\co,V}}%

\global\long\def\Hcm{H_{[\co,m]}}%

\global\long\def\HcmV{H_{[\co,m],V}}%

\global\long\def\Hcmn{H_{[\co,m,n]}}%

\global\long\def\HcmnV{H_{[\co,m,n],V}}%

\global\long\def\Hcmo{H_{[\co,m,1]}}%

\global\long\def\Hcz{H_{[\co,0]}}%

\global\long\def\Hco{H_{[\co,1]}}%

\global\long\def\thc{\theta_{\co}}%

\global\long\def\thcm{\theta_{[\co,m]}}%

\global\long\def\thcmn{\theta_{[\co,m,n]}}%

\global\long\def\pc{p_{\co}}%

\global\long\def\pcm{p_{[\co,m]}}%

\global\long\def\pcmn{p_{[\co,m,n]}}%

\global\long\def\qc{q_{\co}}%

\global\long\def\qcm{q_{[\co,m]}}%

\global\long\def\qcmn{q_{[\co,m,n]}}%

\global\long\def\tc{t_{\co}}%

\global\long\def\tcm{t_{[\co,m]}}%

\global\long\def\tcmn{t_{[\co,m,n]}}%

\global\long\def\tco{t_{[\co,1]}}%

\global\long\def\tcmo{t_{[\co,-1]}}%

\global\long\def\tcz{t_{[\co,0]}}%

\global\long\def\lc{\lambda_{\co}}%

\global\long\def\lcm{\lambda_{[\co,m]}}%

\global\long\def\lcmn{\lambda_{[\co,m,n]}}%

\global\long\def\lo{\lambda_{\mathbf{o}}}%

\global\long\def\mo{\mu_{\mathbf{o}}}%

\global\long\def\Nc{N_{\co}}%

\global\long\def\Ncm{N_{[\co,m]}}%

\global\long\def\Ncmn{N_{[\co,m,n]}}%

\global\long\def\Ic{I_{\co}}%

\global\long\def\Ico{I^{1}_{[\co,1]}}%

\global\long\def\Icmno{I^{1}_{[\co,m,n]}}%

\global\long\def\IcmnN{I^{M+1}_{[\co,m,n]}}%

\global\long\def\Icz{I_{[\co,0]}}%

\global\long\def\Jcz{J_{[\co,0]}}%

\global\long\def\Kcz{K_{[\co,0]}}%

\global\long\def\cz{[\co,0]}%

\global\long\def\sigcz{\sigma_{[\co,0]}}%

\global\long\def\Jcm{J_{[\co,m]}}%

\global\long\def\Kcm{K_{[\co,m]}}%

\global\long\def\Icm{I_{[\co,m]}}%

\global\long\def\Icmn{I_{[\co,m,n]}}%

\global\long\def\Icmo{I_{[\co,m,1]}}%

\global\long\def\Icmoo{I^{1}_{[\co,m,1]}}%

\global\long\def\Icmoi{I^{i}_{[\co,m,1]}}%

\global\long\def\Icmi{I^{i}_{[\co,m]}}%

\global\long\def\Icmni{I^{i}_{[\co,m,n]}}%

\global\long\def\Icmnj{I^{j}_{[\co,m,n]}}%

\global\long\def\Icmii{I^{i+1}_{[\co,m]}}%

\global\long\def\Icmnii{I^{i+1}_{[\co,m,n]}}%

\global\long\def\ohn{\omega_{\frac{p_{k}}{q_{k}}}}%

\global\long\def\ohnm{\omega_{\frac{p_{k-1}}{q_{k-1}}}}%

\global\long\def\ohnp{\omega_{\frac{p_{k+1}}{q_{k+1}}}}%

\global\long\def\fl#1{\emph{\ensuremath{\left\lfloor #1\right\rfloor }}}%

\global\long\def\set#1#2{\left\{  #1\thinspace:\thinspace#2\right\}  }%

\global\long\def\newmacroname{\{\}}%

\global\long\def\crit{V_{\mathrm{crit}}}%

\global\long\def\crito{V^{quasi}_{\mathrm{crit}}}%

\title{The dry ten Martini problem for Sturmian Hamiltonians}
\author{Ram Band, Siegfried Beckus, Raphael Loewy}
\address{Department of Mathematics\\
Technion - Israel Institute of Technology\\
Haifa, Israel}
\email{ramband@technion.ac.il}
\address{Institute of Mathematics\\
University of Potsdam\\
Potsdam, Germany}
\email{beckus@uni-potsdam.de}
\address{Department of Mathematics\\
Technion - Israel Institute of Technology\\
Haifa, Israel}
\email{loewy@technion.ac.il}
\begin{abstract}
The dry ten Martini problem for Sturmian Hamiltonians is solved. Concretely,
we prove that all the predicted spectral gaps ``are there'' for
all the Schrödinger operators with Sturmian potentials and non-vanishing
coupling constant. A key approach towards the solution is a representation
of the spectrum as the boundary of an infinite tree. This tree is
constructed using periodic approximations and encodes substantial
spectral characteristics.
\end{abstract}

\maketitle

\section{Introduction and main results \label{sec: introduction and main results}}

For $\alpha\in[0,1]$ and $V\in\R$, consider the self-adjoint operator
$H_{\alpha,V}:\ell^{2}(\Z)\to\ell^{2}(\Z)$ defined by 
\begin{align}
(\Ham\psi)(n) & :=\psi(n+1)+\psi(n-1)+V\chi_{\left[1-\alpha,1\right)}(n\alpha\mod 1)\thinspace\psi(n),\label{eq: Hamiltonian defined}
\end{align}
where $\chi_{\left[1-\alpha,1\right)}$ is the characteristic function
of the interval $\left[1-\alpha,1\right)$ and $V\in\R$ is the strength
of the potential, which is called the \emph{coupling constant}. When
$\alpha\notin\Q$, this operator $\Ham$ is called a \emph{Sturmian
Hamiltonian,} since the sequence $\chi_{\left[1-\alpha,1\right)}(\xi+n\alpha\mod 1)$
is called a \emph{Sturmian sequence} for $\xi\in[0,1]$. The parameter
$\xi$ may be set to zero for the purpose of the current paper, see
Appendix~\ref{App: Sturmian dynamical systems}.

Let $\Ham|_{[0,n-1]}$ be the restriction of the operator to $\ell^{2}(\{0,\ldots,n-1\})$.
Then $\Ham|_{[0,n-1]}$ is a hermitian $n\times n$ matrix with $\sigma\left(\Ham|_{[0,n-1]}\right)$
denoting its multiset of $n$ eigenvalues (repeated according to their
multiplicities). The limit 
\begin{equation}
\IDS(E):=\lim_{n\to\infty}\frac{\#\set{\lambda\in\sigma\left(\Ham|_{[0,n-1]}\right)}{\lambda\leq E}}{n}\label{eq: definition of DOS}
\end{equation}
is known to exist for all $\alpha\in[0,1]$, $V\in\R$ and $E\in\R$,
see e.g. \cite{Hof93,DaFi22-book_1}. The function $E\mapsto N_{\alpha,V}(E)$
is called the \emph{integrated density of states (IDS)} of $H_{\alpha,V}$.
We denote the spectrum of $\Ham$ by $\sigma(\Ham)$, and mention
two fundamental properties of the IDS:
\begin{enumerate}[label=(IDS\arabic*)]
\item \begin{flushleft}
\label{enu: IDS-property-1} The IDS, $\IDS:\R\rightarrow\left[0,1\right]$
is a monotone, non-decreasing and continuous function.
\par\end{flushleft}
\item \begin{flushleft}
\label{enu: IDS-property-2} We have $E\in\R\backslash\sigma(\Ham)$
if and only if there exists an $\varepsilon>0$ such that the restriction
$\IDS$ is constant on $\left(E-\varepsilon,E+\varepsilon\right)$.
\par\end{flushleft}

\end{enumerate}
The connected components of $\R\backslash\sigma(\Ham)$ are called
\emph{spectral gaps} (or just gaps). Since the IDS is constant on
the spectral gaps, and attains different values at different gaps,
these values are commonly called \emph{gap labels}. Our main theorem
determines the set of appearing gap labels and thus solves the dry
ten Martini problem for Sturmian Hamiltonians.
\begin{thm}
[All gaps are there]\label{thm: all gaps are open} For all $\alpha\in[0,1]\setminus\Q$
and $V\in\R\setminus\{0\}$,
\begin{equation}
\set{\IDS(E)}{E\in\R\backslash\sigma(\Ham)}=\set{l\alpha\mod 1}{l\in\Z}\cup\{1\}.\label{eq: thm-all gaps are open}
\end{equation}
\end{thm}

In the next subsection, we provide a brief historical account of the
dry ten Martini problem. Afterwards, we provide two additional main
theorems, and immediately use them to prove Theorem~\ref{thm: all gaps are open}.

\subsection{The dry ten Martini problem}

``Are all gaps there?'', asked Kac in 1981 during a talk at the
AMS annual meeting, and offered ten Martinis for the solution. This
led Simon \cite{Sim82-review} to coin the names the \emph{Ten Martini
Problem (}\emph{TMP}\emph{)} and the \emph{dry ten Martini problem
(}\emph{DTMP}\emph{)} for two related questions concerning the almost
Mathieu operator (AMO). The first problem, TMP, is whether the AMO
has Cantor spectrum for all irrational frequencies and non-zero coupling
constants. An affirmative answer for the TMP was given by Avila and
Jitomirskaya \cite{AviJit_analsmath09}. Further remarkable results
on Cantor spectrum for generic quasiperiodic Schrödinger operators
are found in \cite{BeSi82,El92,Puig06,AvBoDa09,GolSch11,EicFilGwaLuk_JFA_2022,GeJitYouZho_arXiv23,GeJitYou_arxiv23,DamFilWan_MatNach_2023,DamLen_ETDS_2025}.
Historical overviews on this problem, the route to its resolution
and further important results appear in \cite{MarJit_etds17,Jitomirskaya2019,DaFi24-book_2}.

The DTMP deals with the values that the IDS attains at the spectral
gaps. The gap labelling theorem \cite{Bellis82_Overview,John82,Bell92-Gap,BelBovGhe92,DamFill23-GapLabel,DamFilZha_jst23}
predicts the possible set of values, which the IDS may attain at the
spectral gaps. The predicted gap labels for the AMO are exactly the
ones as for the Sturmian Hamiltonians, see the right hand side of
(\ref{eq: thm-all gaps are open}). The DTMP is whether or not all
these values are attained, or quoting Kac, ``Are all gaps there?''.
We do not exhaustively cover here the literature on the DTMP for the
AMO. A substantial progress towards its solution was achieved in \cite{ChElYu90,Puig04,AvJit10-JEMS,AviBocDam_JEMS_2012,LiYu15}.
The most up to date result appears in \cite{AviYouZho23}, where Avila,
You and Zhou solve the DTMP for the non-critical AMO. A more thorough
historical account on the DTMP for the AMO can be found there. We
refer to \cite{Han_tams18,DamLi_arx24,DamEmiFil_arx24,CedLi_arx25,GeWanXu_arx25}
for more results about existence of spectral gaps for models other
than the AMO.

In the current work, we treat a different class of operators, the
Sturmian Hamiltonians (\ref{eq: Hamiltonian defined}). This model
was introduced and studied in \cite{KohKadTan_prl83,OstKim_phys85}
being a guiding model for one-dimensional quasicrystals. We now describe
the state of the art results for TMP and DTMP for these operators.
A first mathematical study of the spectral properties of the Sturmian
Hamiltonians can be found in Casdagli's paper \cite{Casdagli1986}
that influenced many of the forthcoming works. In \cite{Sut89,BIST89}
it was shown that the spectrum of the Sturmian Hamiltonians is a Cantor
set of Lebesgue measure zero, thus solving the TMP. This was generalized
in \cite{Len02,DaLe06_Boshernitzan,DaLe06_ZeroMeasure} by Damanik
and Lenz for aperiodic Schrödinger operators satisfying the so-called
Boshernitzan condition \cite{Boshernitzan_jam84}. This was also extended
to Jacobi operators in \cite{BePo13}. A significant step towards
the DTMP solution was done by Raymond \cite{Raym95}, who proved (\ref{eq: thm-all gaps are open})
for all $\alpha\not\in\Q$ under the additional assumption that $V>4$.
This unpublished result is part of his thesis \cite{Raym95-thesis}
and will appear in a revised version in \cite{Raym-AperiodicOrder}.
The reader is also referred to \cite{BaBeBiTh22} for a review of
\cite{Raym95}, which is adapted to the conventions of the current
paper. Damanik and Gorodetski \cite{DamanikGorodetski2011} showed
(\ref{eq: thm-all gaps are open}) for the Fibonacci Hamiltonian,
i.e. $\alpha=\frac{\sqrt{5}-1}{2}$, if the coupling constant $V$
is small enough. Mei \cite{Mei14} extended the previous result proving
(\ref{eq: thm-all gaps are open}) for $\alpha\not\in\Q$ with eventually
periodic continued fraction expansion, also in the small coupling
regime. The most recent substantial result was achieved in 2016 by
Damanik, Gorodetski and Yessen. In an extensive study \cite{DaGoYe16},
covering many aspects of the Fibonacci Hamiltonian, they proved that
(\ref{eq: thm-all gaps are open}) holds for $\alpha=\frac{\sqrt{5}-1}{2}$.
The current paper provides the complete affirmative solution of the
DTMP for Sturmian Hamiltonians -- Theorem~\ref{thm: all gaps are open}.
\begin{rem*}
The proof of the DTMP presented here does not rely on the solution
of the TMP. In fact, Theorem~\ref{thm: all gaps are open} together
with standard arguments (see, e.g., \cite[Rem. 5.2]{ChElYu90}) and
the density of the gap labels in $[0,1]$ implies that the spectrum
of the Sturmian Hamiltonian is a Cantor set. Hence, our result also
provides an alternative proof of the Ten Martini Problem for Sturmian
Hamiltonians. 
\end{rem*}

\subsection{The spectra of the periodic (rational) approximations of $\protect\Ham$\label{subsec: Spectra of periodic approximations}}

The first step towards the proof of Theorem~\ref{thm: all gaps are open}
is done by considering the spectra of the periodic (also known as
rational) approximations of $\Ham$, which are introduced next.

The periodic approximations of $\Ham$ are defined via Diophantine
approximations of $\alpha\in\left[0,1\right]\backslash\Q$. Each $\alpha\in\left[0,1\right]\backslash\Q$
is uniquely presented in terms of its continued fraction expansion,
\begin{equation}
\alpha=c_{0}+\frac{1}{c_{1}+\frac{1}{c_{2}+\frac{1}{\ddots}}},\label{eq: infinite continued fraction expansion}
\end{equation}
where $c_{0}=0$ in our case and $c_{n}\in\N$ for all $n\in\N$.
Truncating the expansion above gives finite continued fraction expansions,
\begin{equation}
\alpha_{k}:=c_{0}+\frac{1}{c_{1}+\frac{1}{\ddots+\frac{1}{c_{k}}}}=\frac{p_{k}}{q_{k}},\qquad k\in\Nz,\label{eq: finite continued fraction expansion}
\end{equation}
where $\Nz:=\N\cup\left\{ 0\right\} $ and for $k\in\N$, $p_{k},q_{k}\in\N$
are chosen to be coprime, and by convention we set $\alpha_{0}=\frac{p_{0}}{q_{0}}=\frac{0}{1}$
(as $c_{0}=0$).

This allows to approximate the spectrum of $\Ham$ in terms of spectra
of periodic operators of the form $\Hrat$ (where $p,q$ are coprime).
Such an operator $\Hrat$ is $q$-periodic and hence its spectral
properties are given by the Floquet-Bloch theory.
\begin{prop}
\label{prop: Basic spectral prop periodic} Let $V\in\mathbb{R}\backslash\left\{ 0\right\} $
and $\frac{p}{q}\in[0,1]$ such that $p$ and $q$ are coprime. Then
$\Hrat$ has absolutely continuous spectrum and the spectrum $\sigma(H_{\frac{p}{q},V})$
consists of exactly $q$ connected components, each being a closed
interval.
\end{prop}

These are well-known properties of periodic Schrödinger operators,
see e.g. \cite{Tes00,DaFi24-book_2}. Nevertheless, not every $q$-periodic
operator has a spectrum consisting of exactly $q$ connected components
(in general this is only an upper bound). This is a specific property
of $\Hrat$, see e.g. \cite[Prop. 3.1]{Raym95}, \cite[Prop. 4.1]{BaBeBiTh22}.

We introduce the notation $\left\{ \sigk\right\} _{k\in\Nmo}$, with
$\Nmo:=\N\cup\{-1,0\}$, for the spectra of the periodic approximants,
\begin{equation}
\sigma_{-1}(V):=\R\quad\mathrm{and}\quad\sigma_{k}(V):=\sigma\left(H_{\frac{p_{k}}{q_{k}},V}\right).\label{eq: sigma_rational}
\end{equation}
The auxiliary spectrum $\sigma_{-1}(V)$ seems artificial at first
sight, but its role becomes clearer in the next subsection (see e.g.,
Theorem~\ref{thm: vertices are spectral bands} and the beginning
of its proof).

The following shows that indeed the spectra of the operators $H_{\alpha_{k},V}$
approximate the spectrum of the Sturmian Hamiltonian $\Ham$.
\begin{prop}
\label{prop: Monotonicity and limit of spectral approximants}\cite{Sut87,BIST89,BIT91}
For all $k\in\N$, and $V\in\R$, the following monotonicity property
holds
\[
\sigkp(V)\subseteq\sigk(V)\cup\sigkm(V).
\]
In addition, 
\[
\lim_{k\rightarrow\infty}\left(\sigk(V)\cup\sigkp(V)\right)=\bigcap_{k\in\N}\left(\sigk(V)\cup\sigkp(V)\right)=\sigma(\Ham),
\]
with the limit taken with respect to the Hausdorff metric on compact
subsets of $\R$.
\end{prop}

In addition, we have the following anti symmetric relation, $\sigk(V)=-\sigk(-V)$
\cite[Lem.~4.1]{BaBeLo_spec26}, which allows us to focus on $V>0$.

These spectral approximations, $\sigk(V)$, may be used to define
an ordered (directed) tree graph, $\tree$, whose boundary represents
the spectrum $\sigma(\Ham)$. After introducing this tree graph and
stating its properties, we are able to prove Theorem~\ref{thm: all gaps are open}.

\subsection{The spectral $\alpha$-tree \label{subsec: Spectral Approximants Tree}}

Next, we define the ordered (directed) tree graph, $\tree$. Towards
this, recall basic graph theory terminology. A \emph{directed graph}
$G$ consists of a countable set $\V$, called \emph{vertex set},
and a set $\mathcal{E}\subseteq\V\times\V$, called the \emph{edge
set}. There is an edge from $u\in\V$ to $w\in\V$ if $(u,w)\in\mathcal{E}$.
The underlying undirected graph of $G$ is the graph obtained by replacing
each directed edge $(u,w)\in\mathcal{E}$ by an undirected edge $\{u,w\}$.
An undirected graph without cycles is called a \emph{tree}. A directed
graph is called a directed tree if its underlying undirected graph
is a tree. A \emph{rooted tree }is a tree which has a single vertex
designated as a root. In the following, we consider an \emph{ordered
rooted directed tree}, which is a rooted directed tree with a strict
(i.e., irreflexive) partial order relation, $\prec$, defined on its
vertex set. The order we use in what follows and indicate by $\prec$
is not the order imposed by the edge directions, but a different one,
see Definition~\ref{def: spectral approximants tree}.

Fix $\alpha\in[0,1]\setminus\Q$ and let $\left(c_{k}\right)^{\infty}_{k=0}$
be the coefficients of its continued fraction expansion, (\ref{eq: infinite continued fraction expansion}).
We recursively describe in the following a specific ordered rooted
directed tree, $\tree$, whose edge and vertex sets are denoted by
$\mathcal{E}_{\alpha}$ and $\V_{\alpha}$, correspondingly. Figure~\ref{Fig: TreeData}
accompanies the tree description.

\begin{figure}[hbt]
\includegraphics[scale=0.82]{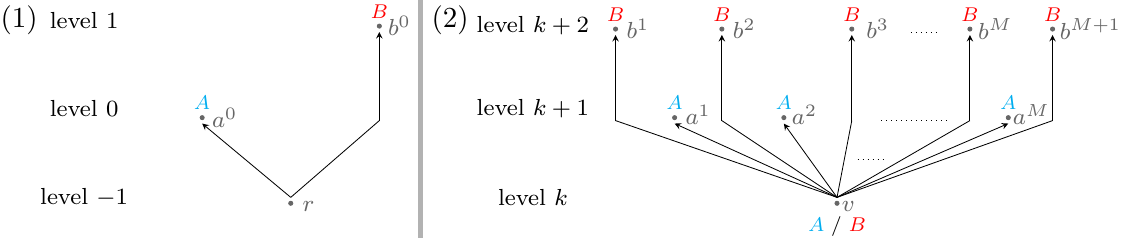}\caption{(1) The root of the tree graph $\protect\tree$ and two adjacent vertices.
(2) A vertex $v$ in level $k$ (for $k\protect\geq0$) and its outgoing
edges to level $k+1$ and $k+2$.\label{Fig: TreeData}}
\end{figure}

We start by designating a single vertex to be the root, $r$. We say
that the root belongs to level $k=-1$ of the tree. Starting from
the root, all other vertices belong to ascending levels $k$ in the
tree and in addition they carry one of the two labels: $A$ or $B$.
There are two vertices to which the root $r$ is connected, $(r,a^{0})\in\mathcal{E}_{\alpha}$
and $(r,b^{0})\in\mathcal{E}_{\alpha}$:
\begin{itemize}
\item We set the vertex $a^{0}$ to be in level $k=0$ and assign $a^{0}$
the label $A$. The vertex $a^{0}$ is the only vertex in level $k=0$.
\item We set the vertex $b^{0}$ to be in level $k=1$ and assign $b^{0}$
the label $B$. Note that there might be other vertices in level $k=1$,
see, e.g. Figure \ref{Fig: TreeData}~(1).
\item These two vertices are ordered $a^{0}\prec b^{0}$.
\end{itemize}
We continue defining the ordered tree $\tree$ recursively. For every
vertex $v$ in level $k$ ($k\geq0$), denote 
\[
M:=\begin{cases}
c_{k+1}-1,\qquad & \text{ if }v\text{ has the label }A,\\
c_{k+1},\qquad & \text{ if }v\text{ has the label }B,
\end{cases}
\]
and (as depicted in Figure \ref{Fig: TreeData}~(2))
\begin{itemize}
\item connect the vertex $v$ to $M$ vertices, $a^{1},\ldots,a^{M}$, all
of which are assigned the label $A$ and they are in level $k+1$,
namely, $(v,a^{i})\in\mathcal{E}_{\alpha}$ for $1\leq i\leq M$.
\item connect the vertex $v$ to $M+1$ vertices, $b^{1},\ldots,b^{M+1}$,
all of which are assigned the label $B$ and they are in level $k+2$,
namely, $(v,b^{j})\in\mathcal{E}_{\alpha}$ for $1\leq j\leq M+1$.
\item These vertices are ordered $b^{1}\prec a^{1}\prec b^{2}\prec\ldots\prec a^{M}\prec b^{M+1}$.
\end{itemize}
\begin{defn}
\label{def: spectral approximants tree} For $\alpha\in[0,1]\setminus\Q$,
the previously described ordered tree, $\tree$, is called the \emph{spectral
$\alpha$-tree}. The following two strict (i.e., irreflexive) partially
order relations are defined on the vertex set $\V_{\alpha}$ of $\tree$:
\begin{itemize}
\item We denote $u\rightarrow w$ whenever there is a directed path connecting
$u$ to $w$.
\item If $u_{1},u_{2}\in\V_{\alpha}$ satisfy $u_{1}\prec u_{2}$, then
we define $w_{1}\prec w_{2}$ for all $w_{1},w_{2}\in\V$ satisfying
$\left(u_{1}\rightarrow w_{1}~\mathrm{or}~u_{1}=w_{1}\right)$ and
$\left(u_{2}\rightarrow w_{2}~\mathrm{or}~u_{2}=w_{2}\right)$.
\end{itemize}
\end{defn}

\begin{rem*}
~
\begin{itemize}
\item We note that the relation $\prec$ is not a total order. But for any
two vertices $u,w\in\V$ with no directed path between them, either
$u\prec w$ or $w\prec u$.
\item We emphasize that the level of a vertex in $\tree$ is not necessarily
its combinatorial distance from the root. This is since the $B$ vertices
are connected by a single edge to a vertex which is two levels below.
\end{itemize}
\end{rem*}
In order to connect the spectral $\alpha$-tree in Definition~\ref{def: spectral approximants tree}
to the spectral approximations, $\sigk(V)$, we introduce the following
conventions. By Proposition~\ref{prop: Basic spectral prop periodic},
for $k\geq0$ and $V\neq0$, the spectrum $\sigk(V)$ consists of
exactly $q_{k}$ intervals (recalling that $\alpha_{k}=\frac{p_{k}}{q_{k}}$).
This leads to the following definition.
\begin{defn}
\label{def: A spectral band is continuous}For $\alpha\in[0,1]\setminus\Q$,
$k\geq0$ and $\alpha_{k}$ as in (\ref{eq: finite continued fraction expansion}).
A map $I:V\mapsto I(V),~V>0,$ is called a \emph{spectral band} in
$\sigk$ if there is a $0\leq j<q_{k}$, such that for all $V>0$,
$I(V)$ is the $j$-th interval (counted from the left) of $\sigk(V)$.
\end{defn}

\begin{rem*}
In the following, we will abuse terminology and also refer to the
evaluation of that map, i.e., $I(V)$, as a spectral band. This is
a common terminology in the literature. Whether a spectral band means
the map itself or its evaluation will be either understood from the
context or explicitly mentioned.
\end{rem*}
Next, we introduce order relations for spectral bands and use these
in Theorem~\ref{thm: vertices are spectral bands} to connect them
to vertices of the ordered tree $\tree$ (see Figure~\ref{fig: basic tree example}
for a demonstration).
\begin{defn}
\label{def: order relations on spectral bands - as maps}Let $I:V\mapsto\left[L(I(V)),R(I(V))\right]$
and $J:V\mapsto\left[L(J(V)),R(J(V))\right]$ be two spectral bands.
We define the following strict (i.e., irreflexive) order relations.
\begin{enumerate}
\item The spectral band \emph{$I$ is }strictly contained \emph{in $J$}:
\[
I\strict J\quad\Leftrightarrow\quad\forall V>0:\quad L(J(V))<L(I(V))\mathrm{<R(I(V))<R(J(V))}.
\]
\item The spectral band \emph{$I$ is} to the left of \emph{$J$ (respectively
$J$} \emph{is} to the right of\emph{ $I$):}
\[
I\prec J\quad\Leftrightarrow\quad\forall V>0:\quad L(I(V))<L(J(V))\mathrm{~and~R(I(V))<R(J(V)).}
\]
\end{enumerate}
Note that it is possible that $I$ is to the left of $J$ even if
$I(V)\cap J(V)\neq\emptyset$ for some value of $V$. We use these
notations also for the evaluation of the spectral bands, i.e., $I(V)\strict J(V)$
and $I(V)\prec J(V)$. Definition~\ref{def: order relations on spectral bands - as maps}
deliberately reuses the notation $\prec$, already introduced in Definition~\ref{def: spectral approximants tree}
for the order relation between vertices of $\tree$. This notational
overlap is intentional and is explained by Theorem~\ref{thm: vertices are spectral bands}\,(\ref{enu: thm-vertices are spectral bands - left-right relation}).
\end{defn}

Following \cite{BaBeLo_spec26} every spectral band has a fixed type
for all $V>0$, which is defined via the relation $\strict$, see
Proposition~\ref{prop: A-B-types-from-algebra-paper} for details.
\begin{defn}
\label{def: A-B-types}Let $\alpha\in[0,1]\backslash\Q$, $V>0$ and
$k\in\Nz$. A spectral band $I(V)$ of $\sigk(V)$ is called
\begin{itemize}
\item \begin{flushleft}
\emph{of type $\tA$} \\
if there exists a spectral band $J(V)$ in $\sigkm(V)$ such that
$\mbox{\ensuremath{I(V)\strict J(V)}}$.
\par\end{flushleft}
\item \begin{flushleft}
\emph{of type $\tB$} \\
if there exists a spectral band $J(V)$ in $\sigkmm(V)$ such that
$\mbox{\ensuremath{I(V)\strict J(V)}}$ and $I(V)\not\subseteq\sigkm(V)$.
\par\end{flushleft}

\end{itemize}
\end{defn}

The connection between the spectral bands of all $\sigk$ and all
the vertices of $\tree$ is given in the following theorem, whose
proof appears in Section~\ref{Sec: proofs of graph related theorems}.
\begin{thm}
\label{thm: vertices are spectral bands}

Let $\alpha\in\left[0,1\right]\backslash\Q$. Let $\tree$ be the
spectral $\alpha$-tree. Then there exists a unique bijection $\Psi$
between the vertices $\mathcal{V}_{\alpha}$ of $\tree$ and all spectral
bands of $\left\{ \sigk\right\} _{k\in\Nmo}$ for $V>0$, such that:
\begin{enumerate}
\item \label{enu: thm-vertices are spectral bands - preserve level k} For
each $k\in\Nmo$, the bijection $\Psi$ maps each vertex in level
$k$ of $\tree$ to a spectral band of $\sigk$.
\item \label{enu: thm-vertices are spectral bands - inclusion} For every
two vertices $u,w$, if $u\to w$ then $\Psi(w)\strict\Psi(u)$.
\item \label{enu: thm-vertices are spectral bands - left-right relation}
If $u_{1},u_{2}$ are vertices in levels $k_{1},k_{2}$ (respectively)
such that $\left|k_{1}-k_{2}\right|\leq1$, then
\[
u_{1}\prec u_{2}\quad\Longleftrightarrow\quad\Psi(u_{1})\prec\Psi(u_{2}).
\]
\item \label{enu: thm-vertices are spectral bands - A-B-types-preserved}A
vertex $u$ is labeled $A$ (respectively $B$) if and only if the
spectral band $\left(\Psi(u)\right)(V)$ is of type $A$ (correspondingly
$B$) for all $V>0$.
\end{enumerate}
\end{thm}

A similar version of Theorem~\ref{thm: vertices are spectral bands}
holds for $V<0$, but for this sake one needs to adjust the definition
of the tree $\tree$ (see discussion in Remark~\ref{rem: tree construction V<0}).
Figure~\ref{fig: basic tree example} demonstrates the bijection
between the graph vertices and the corresponding spectral bands, for
$\tree$ if $\alpha$ has continued fraction expansion $\left(c_{k}\right)^{\infty}_{k=0}$
starting with $0,1,2,3$.

\begin{figure}
\includegraphics[scale=0.9]{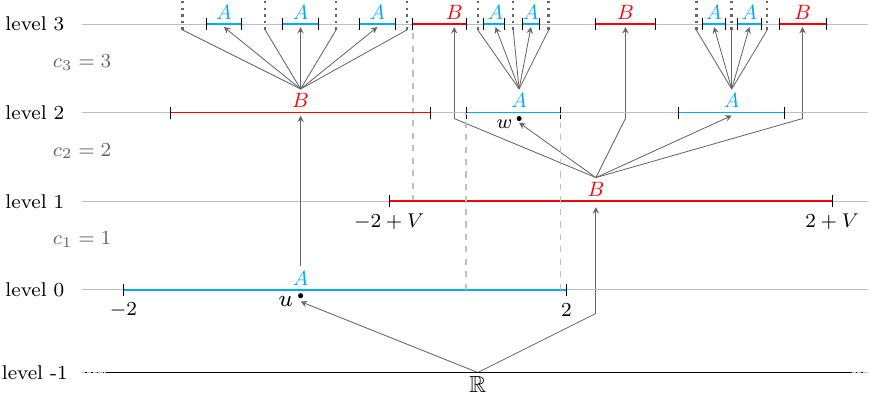}

\caption{\label{fig: basic tree example} An example of a spectral $\alpha$-tree
is sketched if $\alpha$ has continued fraction expansion $\left(c_{k}\right)^{\infty}_{k=0}$
starting with $0,1,2,3$, see Definition~\ref{def: spectral approximants tree}.
The vertices of the graph are drawn as the spectral bands to which
they are mapped by $\Psi$. The two vertices $u,w$ which are marked
satisfy $u\prec w$, but their corresponding spectral bands satisfy
$\Psi(u)\protect\not\prec\Psi(w)$.}
\end{figure}

\begin{example}
\label{exa: left-right order relation only for adjacent levels}Theorem~\ref{thm: vertices are spectral bands}~(\ref{enu: thm-vertices are spectral bands - left-right relation})
claims that $\Psi$ preserves the order relation $\prec$ only for
vertices that are in the same level or in consecutive levels. We note
that this order relation might not be preserved for vertices in levels
which are farther apart. This is demonstrated in Figure~\ref{fig: basic tree example};
the vertices $u$ in level $0$ and $w$ in level $2$ satisfy $u\prec w$
by the order relation defined on $\tree$ but $\Psi(u)\not\prec\Psi(w)$
as sketched in the figure since $\left(\Psi(w)\right)(V)\strict\left(\Psi(u)\right)(V)$
for some values of $V$ (e.g., $V=1$).
\end{example}

\subsection{Connecting $\partial\protect\tree$ with $\sigma(\protect\Ham)$
and proving Theorem~\ref{thm: all gaps are open}~\label{subsec: proof that all gaps are there}}

Theorem~\ref{thm: vertices are spectral bands} connects the spectra
of the approximant operators $H_{\alpha_{k},V}$ with the vertices
of $\tree$. Our next step is to connect the spectrum of $\Ham$ with
the boundary of $\tree$, and use this connection to express the IDS
value. This will allow us to prove Theorem~\ref{thm: all gaps are open}. 

Denote the boundary of $\tree$ by 
\[
\partial\tree:=\set{\gamma=\left(u_{0},u_{1},u_{2}\ldots\right)}{u_{0}\mathrm{~is~the~root~of}~\tree~\textrm{and }(u_{m-1},u_{m})\in\mathcal{E}_{\alpha}~\textrm{for all }m\in\N},
\]
i.e. the set of all infinite paths which start from the root. This
boundary $\partial\tree$ inherits a natural total order from the
partial order $\prec$ on the vertex set $\V$. Specifically, let
$\gamma_{1}=\left(u_{0},u_{1},\ldots\right)$ and $\gamma_{2}=\left(w_{0},w_{1},\ldots\right)$
be in $\partial\tree$. If $\gamma_{1}=\gamma_{2}$, we set $\gamma_{1}\preceq\gamma_{2}$
and $\gamma_{2}\preceq\gamma_{1}$ (so that the order is reflexive).
Otherwise, there exists a unique $k\ge0$ such that $u_{k-1}=w_{k-1}$
and $u_{k}\neq w_{k}$. By construction (see Definition~\ref{def: spectral approximants tree}),
either $u_{m}\prec w_{m}$ for all $m\geq k$ or $w_{m}\prec u_{m}$
for all $m\geq k$. In the former case, we set $\gamma_{1}\preceq\gamma_{2}$
and in the latter case, we set $\gamma_{2}\preceq\gamma_{1}$.

Given an infinite path $\gamma=\left(u_{0},u_{1},\ldots\right)\in\partial\tree$
and $V>0$, Theorem~\ref{thm: vertices are spectral bands}~(\ref{enu: thm-vertices are spectral bands - inclusion})
implies $\Psi(u_{m})\strict\Psi(u_{m-1})$ for all $m\in\N$. Thus,
for all values $V>0$, the intersection $\cap_{m\in\N}\left(\Psi(u_{m})\right)(V)$
of nested compact intervals is non-empty and connected. By Proposition~\ref{prop: Monotonicity and limit of spectral approximants},
this intersection is contained in the spectrum $\sigma(\Ham)$. Furthermore,
the Lebesgue measure of the spectral bands $\Psi(u_{m})$ is bounded
from above by $2\pi/q_{m}$, \cite[Thm. 7.5.1]{DaFi24-book_2}. Hence,
the intersection $\cap_{m\in\N}\left(\Psi(u_{m})\right)(V)$ consists
of a single point. We denote this point by $E_{\alpha}(\gamma;V)$,
i.e., $\cap_{m\in\N}\left(\Psi(u_{m})\right)(V)=\left\{ E_{\alpha}(\gamma;V)\right\} $.
This defines a map
\[
\emap(~\cdot~;V):\partial\tree\rightarrow\sigma(\Ham)
\]
satisfying the following properties.
\begin{thm}
\label{thm: paths and spectra}~ Let $\alpha\in[0,1]\setminus\Q$
and $V>0$.
\begin{enumerate}
\item \label{enu: thm-paths and spectra - bijection}The map $\emap(~\cdot~;V):\partial\tree\rightarrow\sigma(\Ham)$
is a bijection.
\item \label{enu: thm-paths and spectra - order preserving}The map $\emap(~\cdot~;V):\partial\tree\rightarrow\sigma(\Ham)$
is order preserving, i.e. $\gamma_{1}\preceq\gamma_{2}$ implies $\emap(\gamma_{1}~;V)\leq\emap(\gamma_{2}~;V)$.
\item \label{enu: thm-paths and spectra - coninuity in V} For all $\gamma\in\partial\tree$,
the map $\emap(\gamma;~\cdot~):\left(0,\infty\right)\rightarrow\R$
is Lipschitz continuous.
\item \label{enu: thm-paths and spectra - IDOS} There exists a function
$N_{\alpha}:\partial\tree\rightarrow[0,1]$ such that for all $V>0$,
\[
\IDS\left(\emap(\gamma;V)\right)=N_{\alpha}(\gamma).
\]
\item \label{enu: thm-paths and spectra - negative V}We have $\sigma(\Ham)=-\sigma(H_{\alpha,-V})$.
Furthermore, for all $\gamma\in\partial\tree$ and $V<0$, 
\[
N_{\alpha,V}\left(-\emap(\gamma;-V)\right)=1-N_{\alpha}(\gamma).
\]
\end{enumerate}
\end{thm}

Note that the tree graph $\tree$ as well as the function $N_{\alpha}$
are $V$-independent. Furthermore, one can explicitly describe the
function $N_{\alpha}:\partial\tree\rightarrow[0,1]$ by the local
tree structure, see Appendix~\ref{App: IDS-explicit-formula}.
\begin{rem*}
A classification of spectral bands of periodic approximations into
types with respect to their nested structure is an approach initiated
by Casdagli \cite{Casdagli1986}. Raymond substantially developed
it (and called it a coding scheme) in order to resolve the Sturmian
DTMP for $V>4$ \cite{Raym95}. In \cite{BaBeLo_spec26} we established
a modified classification valid for all $V\neq0$. We also refer to
\cite{BaBeBiTh22} for a detailed discussion of the differences. In
contrast to the coding approach of \cite{Raym95}, we adopt here a
graph-based viewpoint. The spectral $\alpha$-tree is $V$-independent
and encodes important spectral properties (as stated in Theroems~\ref{thm: vertices are spectral bands} and \ref{thm: paths and spectra}).
These properties of the spectral $\alpha$-tree are proven with the
aid of new methods and concepts such as admissibility of band edges
and the space of all finite continued fraction expansions (see Sections~\ref{sec: Tools for Injectivity proof}
and \ref{sec: Injectivity-Proof}). Using Theroems~\ref{thm: vertices are spectral bands} and \ref{thm: paths and spectra}
the proof of Theorem~\ref{thm: all gaps are open}, which fully solves
the Sturmian DTMP is rather short and follows next.
\end{rem*}
\begin{proof}
[Proof of Theorem \ref{thm: all gaps are open}]

Let $\alpha\in[0,1]\setminus\Q$ and $V\in\R\setminus\{0\}$. The
inclusion 
\[
\big\{\IDS(E)\,|\,E\in\R\backslash\sigma(H_{\alpha,V})\big\}\subseteq\big\{ l\alpha\mod 1\,|\,l\in\Z\big\}\cup\{1\}
\]
is part of the gap labelling theorem \cite{BelBovGhe92,DamFill23-GapLabel}.
We need only to show the other inclusion. Clearly, the values of the
IDS at the two unbounded spectral gaps are $0$ and $1$. More precisely,
we have $\IDS(E)=0$ for $E<\inf\sigma\left(H_{\alpha,V}\right)$
and $\IDS(E)=1$ for $E>\sup\sigma\left(H_{\alpha,V}\right)$. Thus,
the gap labels $0$ ($l=0$) and $1$ are contained in $\big\{\IDS(E)\,|\,E\in\R\backslash\sigma(H_{\alpha,V})\big\}$.

Let $l\in\Z\setminus\left\{ 0\right\} $ and $V>0$. By \cite{Raym95}
(see also \cite[Thm.~5.25]{BaBeBiTh22}), there exists a $\widetilde{V}>4$
and two different values $\widetilde{E_{1}},\widetilde{E_{2}}\in\sigma(H_{\alpha,\widetilde{V}})$
such that
\[
N_{\alpha,\widetilde{V}}(\widetilde{E_{1}})=N_{\alpha,\widetilde{V}}(\widetilde{E_{2}})=l\alpha\mod 1.
\]
By the surjectivity of the map $\emap(~\cdot~;\widetilde{V})$ (Theorem~\ref{thm: paths and spectra}~(\ref{enu: thm-paths and spectra - bijection})),
we have two different infinite paths $\gamma_{1},\gamma_{2}\in\partial\tree$
such that $\widetilde{E_{1}}=\emap(\gamma_{1};\widetilde{V})$ and
$\widetilde{E_{2}}=\emap(\gamma_{2};\widetilde{V})$.

We use these paths $\gamma_{1}$,$\gamma_{2}$ to designate another
pair of energy values $E_{1}:=\emap(\gamma_{1};V)$ and $E_{2}:=\emap(\gamma_{2};V)$.
By the injectivity of the map $E_{\alpha}(~\cdot~;V)$ (Theorem~\ref{thm: paths and spectra}~(\ref{enu: thm-paths and spectra - bijection})),
we get that $E_{1}\neq E_{2}$. Applying Theorem~\ref{thm: paths and spectra}~(\ref{enu: thm-paths and spectra - IDOS})
yields
\[
\IDS(E_{i})=N_{\alpha}(\gamma_{i})=N_{\alpha,\widetilde{V}}(\widetilde{E_{i}})=l\alpha\mod 1,\qquad i\in\left\{ 1,2\right\} .
\]
Thus, we have identified two different spectral values $E_{1},E_{2}\in\sigma(\Ham)$
such that $\IDS(E_{1})=\IDS(E_{2})=l\alpha\mod 1$. By the monotonicity
of the IDS (see \ref{enu: IDS-property-1}) we get that $\IDS$ is
constant on the interval $\left(E_{1},E_{2}\right)$. By \ref{enu: IDS-property-2},
the interval $\left(E_{1},E_{2}\right)$ is a spectral gap with the
required gap label $l\alpha\mod 1$. We have thus proven the equality
in (\ref{eq: thm-all gaps are open}) for all $V>0$.

If $V<0$, the proof follows similarly as above with the following
slight modifications. Let $l\in\Z$ and $V<0$. By \cite{Raym95}
(see also \cite[Thm.~5.25]{BaBeBiTh22}), there exists a $\widetilde{V}>4$
and two different values $\widetilde{E_{1}},\widetilde{E_{2}}\in\sigma(H_{\alpha,\widetilde{V}})$
such that 
\[
N_{\alpha,\widetilde{V}}(\widetilde{E_{1}})=N_{\alpha,\widetilde{V}}(\widetilde{E_{2}})=(-l)\alpha\mod 1.
\]
Now we proceed as before, defining $\gamma_{1},\gamma_{2}\in\partial\tree$
such that $\widetilde{E_{1}}=\emap(\gamma_{1};\widetilde{V})$ and
$\widetilde{E_{2}}=\emap(\gamma_{2};\widetilde{V})$. Let $E_{1}:=-\emap(\gamma_{1};-V)$
and $E_{2}:=-\emap(\gamma_{2};-V)$, which are different by Theorem~\ref{thm: paths and spectra}~(\ref{enu: thm-paths and spectra - bijection}).
Then Theorem~\ref{thm: paths and spectra}~(\ref{enu: thm-paths and spectra - IDOS})
and (\ref{enu: thm-paths and spectra - negative V}) imply $E_{1},E_{2}\in\sigma(H_{\alpha,V})$
and for $i\in\left\{ 1,2\right\} $,
\[
\IDS(E_{i})=1-N_{\alpha}(\gamma_{i})=1-N_{\alpha,\widetilde{V}}(\widetilde{E_{i}})=1-(-l)\alpha\mod 1=l\alpha\mod 1.
\]
Exactly as above we conclude that $E_{1},E_{2}$ are the edges of
a spectral gap at which the IDS attains the required gap label $l\alpha\mod 1$.
\end{proof}

\begin{rem*}
We note that we have actually proven above that if a gap label appears
for some $V\neq0$ (that is, it is attained by the IDS), then it appears
for all $V\neq0$. Here, we use that all gap labels appear for $V>4$
proven in \cite{Raym95}.
\end{rem*}
The injectivity of the map $E_{\alpha}$ in Theorem~\ref{thm: paths and spectra}
is a crucial ingredient in the previous proof. For $V>4$, it follows
rather straightforwardly since three consecutive spectra cannot intersect,
see Proposition~\ref{prop: traceMaps}. For $V\leq4$, however, spectral
bands overlap, making the proof of injectivity substantially more
difficult. We overcome this difficulty using several new tools and
concepts:
\begin{itemize}
\item a classification of spectral bands for all $V\neq0$ which is independent
of $V>0$ (respectively $V<0$), and a $V$-independent description
via the spectral tree,
\item the space of finite continued fraction expansions (Sections~\ref{sec: Tools for Injectivity proof}
and \ref{sec: Injectivity-Proof}),
\item admissibility of spectral band edges (Subsection~\ref{subsec: admissibility-and-triple-trace-product}).
\end{itemize}
We expect that these methods will also be useful for estimating Hausdorff
dimensions and studying self-similar properties of the Kohmoto butterfly.
A starting point in this direction is the estimate in Remark~\ref{rem: Estimate Coupling - B-band overlap},
which bounds the number of overlapping spectral bands of type $\tB$
in terms of $V$.

\subsection{A bird's-eye view on the spectrum of Sturmian Hamiltonians}

Sturmian Hamiltonians belong to the class of dynamically defined Schrödinger
operators. Various characterizations of their spectra have been developed
in the literature, and these have proved useful in their spectral
analysis. For detailed accounts, we refer the reader to \cite{DaEmGo15-survey,Dam17-Survey,DaFi22-book_1,DaFi24-book_2};
here we recall only the main tools used in the analysis of Sturmian
Hamiltonians. Their spectrum admits several alternative characterizations:
\begin{itemize}
\item The energies for which the Lyapunov exponent vanishes (denoted by
$\mathcal{Z}$) are used to prove Cantor spectrum of Lebesgue measure
zero \cite{Sut89,BIST89,Len02,DaLe06_Boshernitzan,DaLe06_ZeroMeasure}
 by applying Kotani theory \cite{Kotani89}.
\item The energies for which the positive semiorbit under the trace map
stays bounded (denoted by $\mathcal{B}$) are used to estimate the
fractal dimension of the spectrum \cite{Casdagli1986,Raym95,Damanik2008,DaGoYe16}.
\item The energies described in terms of a coding scheme (denoted here by
$\Pi$). This is an approach influenced by Casdagli \cite{Casdagli1986}
and fully developed by Raymond \cite{Raym95} to show that all the
gaps are there for $V>4$. This coding scheme turns also to be useful
for studying the transport exponents and the fractal dimensions of
the spectrum \cite{Raym95,KiKiLa03,Damanik2008,DamanikGorodetski2011,Liu2014,Damanik2015,CaoQu_arXiv23}.
Similar coding scheme appears also for the Period doubling sequence
\cite{LiQuYa22-PeriodDoubl}.
\end{itemize}
In the present paper, we shift the viewpoint from this coding scheme,
which applies for $V>4$, to the spectral $\alpha$-tree, its boundary
$\partial\tree$ and the map $\emap$. We discussed the relationship
between these two viewpoints in \cite{BaBeBiTh22}. With this additional
perspective, the different representations of the spectrum may be
summarized as
\[
\sigma(\Ham)=\mathcal{B}=\mathcal{Z}=\Pi=\emap\left(\partial\tree;V\right).
\]
This last perspective provides a substantial step toward the solution
of the dry ten Martini problem for Sturmian Hamiltonians.

This paper is organized as follows. In Section~\ref{Sec: proofs of graph related theorems},
we prove Theorem~\ref{thm: vertices are spectral bands} and Theorem~\ref{thm: paths and spectra}
except for the injectivity of the map $\emap(~\cdot~;V):\partial\tree\rightarrow\sigma(\Ham)$.
In Section~\ref{sec: Tools for Injectivity proof}, we develop some
tools and methods, which are then used in Section~\ref{sec: Injectivity-Proof}
to establish the injectivity of $E_{\alpha}$.

\section{The spectral $\alpha$-tree and its connection to the Integrated
Density of States \label{Sec: proofs of graph related theorems}}

The initial ingredients for the proofs of Theorem~\ref{thm: vertices are spectral bands}
and Theorem~\ref{thm: paths and spectra} were developed in our preceding\footnote{This article originates from the manuscript first posted as arXiv:2402.16703v1,
which has since been split into two parts. The part needed for proving
Proposition~\ref{prop: A-B-types-from-algebra-paper} now appears
as \cite{BaBeLo_spec26}.} work \cite{BaBeLo_spec26} and are recalled next. Building on these
results, we prove in this section Theorems~\ref{thm: vertices are spectral bands} and \ref{thm: paths and spectra},
except for the injectivity of $E_{\alpha}$. 
\begin{prop}
\label{prop: A-B-types-from-algebra-paper}\cite{BaBeLo_spec26} Let
$\alpha\in[0,1]\backslash\Q$, $V>0$ and let $k\in\N$. Every spectral
band of $\sigk(V)$ is of type $\tA$ or $\tB$ and its type is independent
of the value of $V>0$.

Moreover, let $I:V\mapsto I(V)$, $V>0$, be a spectral band of $\sigk$
and denote
\[
M:=\begin{cases}
c_{k+1}-1 & I~\textrm{is of type }\tA\\
c_{k+1} & I~\textrm{is of type }\tB
\end{cases}.
\]
The spectral band $I$ satisfies the following properties:
\begin{enumerate}
\item \label{enu: forward-A} There exist exactly $M$ spectral bands $A^{1},\ldots,A^{M}$
in $\sigkp$ which are all strictly contained in $I$. In particular,
these spectral bands are all of type $\tA$ for all $V>0$.
\item \label{enu: forward-B} There exist exactly $M+1$ spectral bands
$B^{1},\ldots,B^{M+1}$ in $\sigkpp$ which are all strictly contained
in $I$. In particular, these spectral bands are all of type $\tB$
for all $V>0$.
\item \label{enu: interlacing} The two sets of spectral bands mentioned
above interlace, i.e., 
\[
B^{1}\prec A^{1}\prec B^{2}\prec A^{2}\prec\ldots\prec A^{M}\prec B^{M+1}\quad\mathrm{for\,all\,}V>0.
\]
\end{enumerate}
\end{prop}

\begin{proof}
By \cite[Prop.~3.4]{BaBeLo_spec26}, our definition of $\tA$ and
$\tB$ types (Definition~\ref{def: order relations on spectral bands - as maps})
is equivalent to the one in \cite{BaBeLo_spec26}. Therefore \cite[Thm.~1.1]{BaBeLo_spec26}
applies here and yields that each spectral band in $\sigk(V)$ is
of type $\tA$ or $\tB$ and its type is independent of the value
of $V>0$. The same theorem also implies the existence part of the
spectral bands in properties (\ref{enu: forward-A}),(\ref{enu: forward-B}),(\ref{enu: interlacing}).
The uniqueness part in properties (\ref{enu: forward-A}) and (\ref{enu: forward-B})
(which is reflected by the word `exactly') is not part of that theorem,
but it follows from \cite[Cor.~3.6]{BaBeLo_spec26}.
\end{proof}

\subsection{Proof of Theorem~\ref{thm: vertices are spectral bands}}

Theorem~\ref{thm: vertices are spectral bands} connects the spectra
of all the rational approximants $\sigk$ to all the vertices of the
spectral $\alpha$-tree. 

\begin{proof}[Proof of Theorem~\ref{thm: vertices are spectral bands}]
 We will omit in the following the $V$ dependence unless it is needed.
Let $\alpha\in\left[0,1\right]\backslash\Q$ with an infinite continued
fraction expansion $\left(0,c_{1},c_{2},c_{3},\ldots\right)$. For
each $k\in\Nz$, denote 
\[
\alpha_{k}:=c_{0}+\frac{1}{c_{1}+\frac{1}{\ddots+\frac{1}{c_{k}}}}=\frac{p_{k}}{q_{k}}
\]
 with coprime $p_{k},q_{k}$, as in (\ref{eq: finite continued fraction expansion}).
By standard properties of continued fractions \cite[Thm.~1]{Khinchin_book64},
we have for $k\in\N$,
\begin{equation}
q_{-1}=0,~q_{0}=1,~q_{k}=c_{k}q_{k-1}+q_{k-2},\qquad k\in\N.\label{eq: Proof vertices are spectral bands - recursion q_k}
\end{equation}
Let $\tree$ be the spectral $\alpha$-tree with edge set $\mathcal{E}_{\alpha}$.
Observe that any bijection satisfying the properties (\ref{enu: thm-vertices are spectral bands - preserve level k})
and (\ref{enu: thm-vertices are spectral bands - left-right relation})
of the theorem must be unique since the vertices within a level $k$
are totally ordered by $\prec$, and correspondingly the spectral
bands within $\sigk$ are totally ordered by $\prec$.

We start by constructing the bijection $\Psi$ between all vertices
of $\tree$ and the spectral bands of all $\left\{ \sigk\right\} _{k\in\Nmo}$
(recalling that by convention $\sigma_{-1}(V)=\R$). We do so inductively
in $k$ and start by handling the levels $k\in\{-1,0,1\}$ in the
induction base, sketched in Figure~\ref{Fig: Proof of vertices are spectral bands - Construction Psi}
(3).

\paragraph*{\uline{Induction base}}

Level $k=-1$ of the graph $\tree$ contains only its root. The corresponding
spectrum is $\sigma_{-1}=\R$. We set $\Psi$ to map the root of the
graph to the ``spectral band'' $\R$ in $\sigma_{-1}$.

In level $k=0$, the graph $\tree$ has a single vertex (connected
by an edge to the root), which has the label $A$ (it is the vertex
$a^{0}$ in Figure~\ref{Fig: Proof of vertices are spectral bands - Construction Psi}~(3)).
It is easy to compute that the corresponding spectrum $\sigma_{0}$
consists of a single spectral band $I_{0}=[-2,2]$ (see also \cite[Lem.~7.4]{BaBeLo_spec26})
and $I_{0}\strict\sigma_{-1}=\R$, so $I_{0}$ is of type $A$ by
Definition~\ref{def: A-B-types}. In level $k=1$, the graph $\tree$
has $c_{1}$ vertices which are totally ordered by the order relation
$\prec$. The rightmost vertex has the label $B$ (it is the vertex
$b^{0}$ in Figure~\ref{Fig: Proof of vertices are spectral bands - Construction Psi}~(3)).
This vertex is directly connected to the root. By construction, the
left-most $c_{1}-1$ vertices $a^{1}\prec\ldots\prec a^{c_{1}-1}$
are labeled $A$ and they are directly connected to the vertex $a^{0}$
in level $k=0$. One can compute that the corresponding spectrum $\sigma_{1}$
consists of $c_{1}$ spectral bands (see details in \cite[Lem.~7.4]{BaBeLo_spec26}).
The left-most $c_{1}-1$ spectral bands of $\sigma_{1}$, denoted
$A^{1},\ldots,A^{c_{1}-1}$, are all strictly contained in $I_{0}$
so by Definition~\ref{def: A-B-types} they are of type $\tA$ (these
are actually the bands mentioned in Proposition~\ref{prop: A-B-types-from-algebra-paper}
(\ref{enu: forward-A})). We denote the rightmost spectral band of
$\sigma_{1}$ by $K_{1}$, i.e., $A^{1}\prec\ldots\prec A^{c_{1}-1}\prec K_{1}$.
A computation (see details in \cite[Lem.~7.4]{BaBeLo_spec26}) shows
that $I_{0}\prec K_{1}\strict\R$, namely $K_{1}$ is strictly contained
in $\sigma_{-1}$ but not in $\sigma_{0}$. Therefore, $K_{1}$ is
of type $\tB$ by Definition~\ref{def: A-B-types}. Using the classification
above, we define $\Psi$ on level $k=1$ by $\Psi(a^{j})=A^{j}$ and
$\Psi(b^{0})=K_{1}$, sketched in Figure~\ref{Fig: Proof of vertices are spectral bands - Construction Psi}~(3).
By construction, $\Psi$ satisfies all the claimed properties (\ref{enu: thm-vertices are spectral bands - preserve level k}),
(\ref{enu: thm-vertices are spectral bands - inclusion}), (\ref{enu: thm-vertices are spectral bands - left-right relation}),
(\ref{enu: thm-vertices are spectral bands - A-B-types-preserved})
of Theorem~\ref{thm: vertices are spectral bands}, for the levels
$k\in\{-1,0,1\}$. We continue recursively constructing the map $\Psi$
and inductively proving that all properties of the theorem hold for
all $k>1$.
\begin{figure}
\includegraphics[scale=1.1]{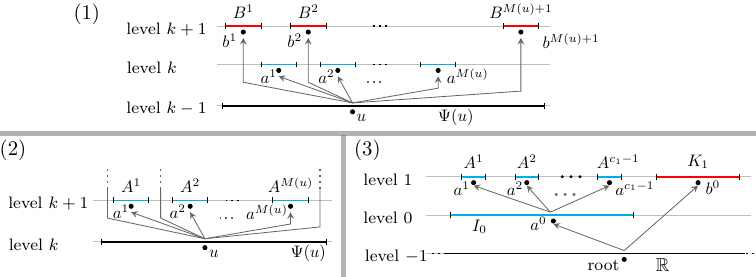}\caption{In (1) and (2), the recursive construction of $\Psi$ is sketched.
In (3), the map $\Psi$ is sketched for the levels $k\in\left\{ -1,0,1\right\} $.
The plotted vertices are mapped to the spectral bands next to them.
}
\label{Fig: Proof of vertices are spectral bands - Construction Psi}
\end{figure}

\uline{Induction step for constructing \mbox{$\Psi$} and proving
its properties}

Let $k\geq1$ be such that there is a map $\Psi$ satisfying up to
(and including) level $k$ all the claimed properties (\ref{enu: thm-vertices are spectral bands - preserve level k}),
(\ref{enu: thm-vertices are spectral bands - inclusion}), (\ref{enu: thm-vertices are spectral bands - left-right relation}),
(\ref{enu: thm-vertices are spectral bands - A-B-types-preserved})
of Theorem~\ref{thm: vertices are spectral bands}. We show that
we can extend $\Psi$ to a map satisfying these properties also in
level $k+1$.

For each vertex $w$ in level $k+1$ with label $A$, there is a vertex
$u$ in $k$ such that there is an edge from $u$ to $w$, i.e. $(u,w)\in\mathcal{E}_{\alpha}$.
The associated spectral band $\Psi(u)$ in $\sigk$ has (by the induction
hypothesis) the label of $u$. Define
\begin{equation}
M(u):=\begin{cases}
c_{k+1}-1 & u\textrm{ is of label }A,\\
c_{k+1} & u\textrm{ is of label }B.
\end{cases}\label{eq: proof of prop: vertices are spectral bands - M notation}
\end{equation}
By the definition of $\tree$ (see Section~\ref{subsec: Spectral Approximants Tree}),
$u$ is connected to $M(u)$ vertices $a^{1}\prec\ldots\prec a^{M(u)}$
of label $A$ in level $k+1$. Similarly, by Proposition~\ref{prop: A-B-types-from-algebra-paper}
(\ref{enu: forward-A}), $\Psi(u)\subseteq\sigk$ strictly contains
exactly $M(u)$ spectral bands $A^{1}\prec\ldots\prec A^{M(u)}$ of
type $A$ in $\sigkp$. We define $\Psi$ to map these $M(u)$ vertices
to these $M(u)$ spectral bands such that the $\prec$ order between
them is preserved (i.e. $\Psi(a^{j})=A^{j}$), see a sketch in Figure~\ref{Fig: Proof of vertices are spectral bands - Construction Psi}~(2).
We repeat this for all vertices $u$ in level $k$. Doing so, we have
bijectively mapped all $A$ labeled vertices in level $k+1$ to all
spectral bands in $\sigkp$ of type $A$. The surjectivity of the
map $\Psi$ in level $k+1$ is guaranteed since Proposition~\ref{prop: A-B-types-from-algebra-paper}
(\ref{enu: forward-A}) ensures that every $\tA$-type spectral band
in $\sigkp$ is one of the spectral bands $\{A^{j}\}^{M(u)}_{j=1}$
which are contained in a spectral band of $\sigk$.

The vertices with label $B$ are treated similarly: For each vertex
$w$ in level $k+1$ with label $B$, there is a vertex $u$ in $k-1$
such that there is an edge from $u$ to $w$. Then the associated
spectral band $\Psi(u)$ in $\sigkm$ has (by the induction hypothesis)
the label of $u$. As before $u$ is connected to $M(u)+1$ vertices
$b^{1}\prec\ldots\prec b^{M(u)+1}$ with label $B$ and $\Psi(u)\subseteq\sigkm$
strictly contains exactly $M(u)+1$ spectral bands $B^{1}\prec\ldots\prec B^{M(u)+1}$
of type $B$ in $\sigkp$. We define $\Psi$ to map these $M(u)+1$
vertices to these $M(u)+1$ spectral bands such that the $\prec$
order between them is preserved (i.e. $\Psi(b^{j})=B^{j}$), see a
sketch in Figure~\ref{Fig: Proof of vertices are spectral bands - Construction Psi}~(1).
We repeat this for all vertices $u$ in level $k-1$. By construction
$\Psi$ maps bijectively all $B$ labeled vertices in level $k+1$
to all spectral bands in $\sigkp$ of type $B$. The surjectivity
here follows exactly as the surjectivity argument written above for
the type $\tA$ bands.

We emphasize at this point that by construction and the interlacing
property (Proposition~\ref{prop: A-B-types-from-algebra-paper}~(\ref{enu: interlacing}))
of the spectral band $\Psi(u)$, we have
\begin{equation}
\Psi(b^{1})\prec\Psi(a^{1})\prec\ldots\prec\Psi(b^{M(u)})\prec\Psi(a^{M(u)})\prec\Psi(b^{M(u)+1}),\label{eq: proof of prop: vertices are spectral bands - interlacing}
\end{equation}
borrowing the notation of Figure~\ref{Fig: Proof of vertices are spectral bands - Construction Psi}~(1).

We conclude that $\Psi$ satisfies the properties (\ref{enu: thm-vertices are spectral bands - preserve level k})
and (\ref{enu: thm-vertices are spectral bands - A-B-types-preserved})
of the theorem. Property (\ref{enu: thm-vertices are spectral bands - inclusion})
follows also immediately from the construction and the induction hypothesis. 

By induction, it remains to prove that if $k\geq1$ and $\Psi$ satisfies
(\ref{enu: thm-vertices are spectral bands - left-right relation})
for all vertices up to (and including) level $k$, then it also satisfies
(\ref{enu: thm-vertices are spectral bands - left-right relation})
for all vertices up to and including level $k+1$. Explicitly, let
$w$ and $\tilde{w}$ be two different vertices, which are either
both in level $k+1$, or such that one of them is in level $k$ and
the other in level $k+1$. We need to show that
\begin{equation}
w\prec\tilde{w}\quad\Longleftrightarrow\quad\Psi(w)\prec\Psi(\tilde{w}).\label{eq: equivalence of prec}
\end{equation}
First, note that we may assume that $w$ and $\tilde{w}$ are not
connected by an edge. Otherwise, by (\ref{enu: thm-vertices are spectral bands - inclusion}),
$w\rightarrow\tilde{w}$ would imply $\Psi(\tilde{w})\strict\Psi(w)$
so that $\Psi(\tilde{w})\not\prec\Psi(w)$ and $\Psi(w)\not\prec\Psi(\tilde{w})$.
But $w\rightarrow\tilde{w}$ implies also $w\not\prec\tilde{w}$ and
$\tilde{w}\not\prec w$, and so the equivalence (\ref{eq: equivalence of prec})
is valid in this case.

We show next that either $w\prec\tilde{w}$ or $\tilde{w}\prec w$
(see also the remark after Definition~\ref{def: spectral approximants tree}).
To see this, choose a path $\gamma$ from the root to $w$ and a path
$\tilde{\gamma}$ from the root to $\tilde{w}$ (both paths are unique).
We denote by $v$ the vertex of maximal level which appears in both
paths (such a vertex exists since the root appears in both paths).
We denote by $u$ the vertex to which $v$ is connected in $\gamma$
and denote by $\tilde{u}$ the vertex to which $v$ is connected in
$\tilde{\gamma}$. By the construction of the tree and the order relation
$\prec$ on its vertices (Definition~\ref{def: spectral approximants tree}
and Figure~\ref{Fig: TreeData}) we have that either $u\prec\tilde{u}$
or $\tilde{u}\prec u$. Since $u\rightarrow w$ and $\tilde{u}\rightarrow\tilde{w}$
we get (by Definition~\ref{def: spectral approximants tree}) that
either $w\prec\tilde{w}$ or $\tilde{w}\prec w$.

We will now show that $w\prec\tilde{w}$ implies $\Psi(w)\prec\Psi(\tilde{w})$.
This actually proves (\ref{eq: equivalence of prec}) which can be
seen as follows. Assume by contradiction that $\Psi(w)\prec\Psi(\tilde{w})$
and $\tilde{w}\prec w$ hold. Then the previous implication leads
to $\Psi(\tilde{w})\prec\Psi(w)$ contradicting $\Psi(w)\prec\Psi(\tilde{w})$.

Suppose $w\prec\tilde{w}$ holds. We now go over all the possible
configurations in which $w\prec\tilde{w}$ and they are either both
in level $k+1$, or one of them is in level $k$ and the other in
level $k+1$. There are 11 such cases (as depicted in Figure~\ref{Fig: Proof of vertices are spectral bands - order preserving})
and we verify that $\Psi(w)\prec\Psi(\tilde{w})$ holds in all these
cases.

We start by checking the cases in which $w$ and $\tilde{w}$ are
different vertices in level $k+1$. Since $w\neq\tilde{w}$, the injectivity
of $\Psi$ implies $\Psi(w)\neq\Psi(\tilde{w})$. Thus, $\left(\Psi(w)\right)(V)\cap\left(\Psi(\tilde{w})\right)(V)=\emptyset$
follows from Proposition~\ref{prop: Basic spectral prop periodic}
(different spectral bands in the same level do not touch). Hence,
we have 
\begin{equation}
\textrm{either }\quad\Psi(w)\prec\Psi(\tilde{w})\quad\textrm{ or }\quad\Psi(\tilde{w})\prec\Psi(w).\label{eq: proof of prop: vertices are spectral bands - prec}
\end{equation}
 By the construction of the tree, there are vertices $u$ and $\tilde{u}$
in level $k$ or $k-1$, which are connected to $w$ and $\tilde{w}$,
i.e., $(u,w),(\tilde{u},\tilde{w})\in\mathcal{E}_{\alpha}$. Then
property (\ref{enu: thm-vertices are spectral bands - inclusion})
implies $\Psi(\tilde{w})\strict\Psi(\tilde{u})$ and $\Psi(w)\strict\Psi(u)$.
To show $\Psi(w)\prec\Psi(\tilde{w})$, one needs to treat 8 different
cases. These cases are plotted in Figure~\ref{Fig: Proof of vertices are spectral bands - order preserving},
(1) to (8) and analyzed below. 
\begin{figure}
\includegraphics[scale=0.98]{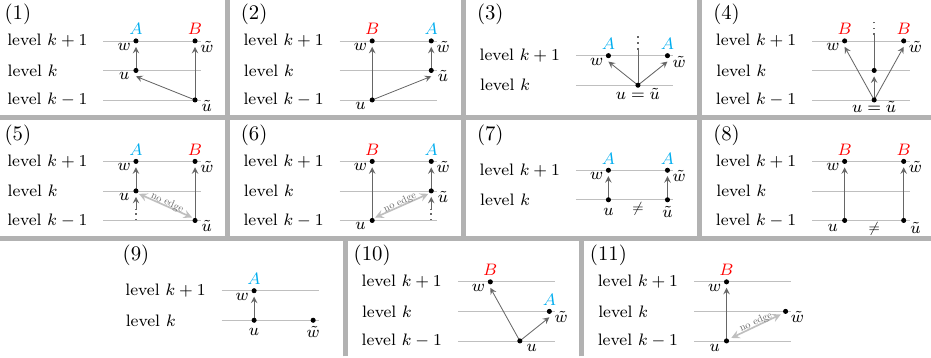}\caption{Outlining the different cases in the proof of property (\ref{enu: thm-vertices are spectral bands - left-right relation})
in the induction step.}
\label{Fig: Proof of vertices are spectral bands - order preserving}
\end{figure}

In case (1), $u\prec\tilde{w}$ follows from $w\prec\tilde{w}$. Recall
that $\Psi$ is defined by sending the vertices which are connected
to $\tilde{u}$ to the spectral bands which are contained in $\Psi(\tilde{u})$
according to Proposition~\ref{prop: A-B-types-from-algebra-paper}~(\ref{enu: forward-A}),(\ref{enu: forward-B}).
In particular, the map $\Psi$ preserves the interlacing property
(Proposition~\ref{prop: A-B-types-from-algebra-paper}~(\ref{enu: interlacing}))
as is indicated in Equation~(\ref{eq: proof of prop: vertices are spectral bands - interlacing}))
which means in our case that $\Psi(u)\prec\Psi(\tilde{w})$. By $\Psi(w)\strict\Psi(u)$
and (\ref{eq: proof of prop: vertices are spectral bands - prec}),
we conclude $\Psi(w)\prec\Psi(\tilde{w})$. The case (2) is treated
similarly. The cases (3) and (4), also follow directly from the recursive
definition of $\Psi$.

In the cases (5) to (8), there is neither an edge between $u$ and
$\tilde{u}$ nor these vertices coincide. By the definition of the
order on the tree $\tree$, we conclude $u\prec\tilde{u}$ since $w\prec\tilde{w}$,
$(u,w)\in\mathcal{E}_{\alpha}$ and $(\text{\ensuremath{\tilde{u},\tilde{w})}\ensuremath{\in\mathcal{E}_{\alpha}}}$.
Since in these cases $u$ and $\tilde{u}$ are in level $k$ or $k-1$,
we get by the induction hypothesis that $u\prec\tilde{u}$ implies
$\Psi(u)\prec\Psi(\tilde{u})$. In the cases (7) and (8), since $u,\tilde{u}$
are in the same level, so are $\Psi(u),\Psi(\tilde{u})$ and hence
we have even $\Psi(u)(V)\cap\Psi(\tilde{u})(V)=\emptyset$ for all
$V>0$. Since $\Psi(w)\strict\Psi(u)$ and $\Psi(\tilde{w})\strict\Psi(\tilde{u})$,
we conclude $\Psi(w)\prec\Psi(\tilde{w})$.  However, in cases (5)
and (6), we cannot directly conclude from $\Psi(u)\prec\Psi(\tilde{u})$,
$\Psi(\tilde{w})\strict\Psi(\tilde{u})$ and $\Psi(w)\strict\Psi(u)$
that $\Psi(w)\prec\Psi(\tilde{w})$. A-priori, it is possible in these
conditions that $\Psi(\tilde{w})\prec\Psi(w)$ if $\Psi(w)$ and $\Psi(\tilde{w})$
are both contained in $\Psi(u)\cap\Psi(\tilde{u})$. However, this
leads to a contradiction as we explain now. Either $\tilde{w}$ is
of label $\tB$ (in case (5)) or $w$ is of label $\tB$ (in case
(6)). Since $\Psi$ preserves the labels, either $\Psi(\tilde{w})$
(in case (5)) or $\Psi(w)$ (in case (6)) is of type $B$. A spectral
band which is of type $\tB$ of $\sigma_{k+1}$ cannot be contained
in a spectral band of $\sigma_{k}$ (Definition~\ref{def: A-B-types}).
This contradicts that both $\Psi(\tilde{w})$ and $\Psi(w)$ are both
contained in $\Psi(u)\cap\Psi(\tilde{u})$. This finishes the proof
for the cases (5) and (6).

Next, let $w$ be a vertex in level $k+1$ and $\tilde{w}$ be a vertex
in level $k$ (which are cases (9),(10),(11)). The symmetric case
of $w$ being in level $k$ and $\tilde{w}$ being in level $k+1$
follows similarly. Recall that we want to show that $w\prec\tilde{w}$
implies $\Psi(w)\prec\Psi(\tilde{w})$. Let $u$ be the vertex in
level $k$ (if $w$ has label $A$) or $k-1$ (if $w$ has label $B$)
such that there is an edge $(u,w)\in\mathcal{E}_{\alpha}$.

If $w$ has label $A$, then $u\neq\tilde{w}$ follows since we assume
that $w,\tilde{w}$ are not connected by an edge. Since $w\prec\tilde{w}$,
we conclude $u\prec\tilde{w}$, which are both in level $k$, see
Figure~\ref{Fig: Proof of vertices are spectral bands - order preserving}~(9).
Thus, the induction hypothesis asserts that $u\prec\tilde{w}$ implies
$\Psi(u)\prec\Psi(\tilde{w})$. Since $\Psi(u)$ and $\Psi(\tilde{w})$
are both spectral bands in $\sigma_{\alpha_{k}}$, they cannot intersect,
i.e. $\left(\Psi(u)\right)(V)\cap\left(\Psi(\tilde{w})\right)(V)=\emptyset$
for all $V>0$. Since $(u,w)\in\mathcal{E}_{\alpha}$ implies $\Psi(w)\strict\Psi(u)$
by property (\ref{enu: thm-vertices are spectral bands - inclusion}),
the previous considerations imply $\Psi(w)\prec\Psi(\tilde{w})$.
This proves case (9).

If $w$ has label $B$, then we have two cases sketched in Figure~\ref{Fig: Proof of vertices are spectral bands - order preserving}
(10) and (11). In case (10), $\Psi(w)\prec\Psi(\tilde{w})$ follows
from $w\prec\tilde{w}$ since $\Psi$ preserves the interlacing property,
(\ref{eq: proof of prop: vertices are spectral bands - interlacing}),
of the spectral band $\Psi(u)$ by construction. In case (11), there
is no edge between $u$ and $\tilde{w}$. By the definition of the
order on the tree $\tree$, we conclude $u\prec\tilde{w}$ since $w\prec\tilde{w}$
and $(u,w)\in\mathcal{E}_{\alpha}$. Thus, the induction hypothesis
($u$ and $\tilde{w}$ are in level $k$ and $k-1$) asserts that
$u\prec\tilde{w}$ implies $\Psi(u)\prec\Psi(\tilde{w})$. Thus, either
$\Psi(w)\subseteq\Psi(\tilde{w})$ or $\Psi(w)\prec\Psi(\tilde{w})$.
However, $\Psi(w)$ is of type $B$ since $w$ has label $B$ and
$\Psi$ preserves the labels. Thus, the spectral band $\Psi(w)$ of
$\sigma_{k+1}$ cannot be contained in a spectral band of $\sigma_{k}$
by Definition~\ref{def: A-B-types}. In particular, $\Psi(w)$ is
not contained in $\Psi(\tilde{w})$ implying $\Psi(w)\prec\Psi(\tilde{w})$
by the previous considerations.
\end{proof}

\subsection{Proof of Theorem \ref{thm: paths and spectra} without the injectivity
of $\protect\emap$}

We first recall the terminology used in Theorem~\ref{thm: paths and spectra}.
The boundary $\partial\tree$ of the ordered tree $\tree$ with edge
set $\mathcal{E}_{\alpha}$ is 
\[
\partial\tree:=\set{\gamma=\left(u_{0},u_{1},u_{2}\ldots\right)}{u_{0}\mathrm{~is~the~root~of}~\tree~\textrm{and }(u_{m-1},u_{m})\in\mathcal{E}_{\alpha}~\textrm{for all }m\in\N},
\]
i.e. the set of all infinite paths which start from the root. Given
$\gamma\in\partial\tree$, we consider the (infinite) intersection
of all spectral bands which correspond to the vertices of $\gamma$,
i.e., $\cap_{m\in\N}\left(\Psi(u_{m})\right)(V)$. By Theorem~\ref{thm: vertices are spectral bands}~(\ref{enu: thm-vertices are spectral bands - inclusion}),
we have $\left(\Psi(u_{m+1})\right)(V)\strict\left(\Psi(u_{m})\right)(V)$.
With this at hand, we have argued in Section~\ref{subsec: proof that all gaps are there}
that this intersection $\cap_{m\in\N}\left(\Psi(u_{m})\right)(V)$
contains a single point in $\sigma(\Ham)$, which we denoted by $E_{\alpha}(\gamma;V)$
(see also \cite[Lem.~5.11]{BaBeBiTh22} for the case $V>4$). This
defines the map 
\[
\emap(~\cdot~;V):\partial\tree\rightarrow\sigma(\Ham),
\]
whose properties are in the focus of Theorem~\ref{thm: paths and spectra}.

We split the statement of Theorem~\ref{thm: paths and spectra} into
various lemmas: Lemmas~\ref{lem: Lipschitz continuous limiting spectrum},
\ref{lem: IDOS is independent of V}, \ref{lem: map boundary-spectrum is surjective},
\ref{lem: order preserving E_=00005Calpha} are proven next in this
section, and Lemma~\ref{lem: map boundary-spectrum is injective}
(establishing the injectivity of $E_{\alpha}$) is proven in Section~\ref{sec: Injectivity-Proof}.
\begin{lem}
\label{lem: Lipschitz continuous limiting spectrum} [also Theorem \ref{thm: paths and spectra}~(\ref{enu: thm-paths and spectra - coninuity in V})]
Let $\alpha\in[0,1]\setminus\Q$. For all $\gamma\in\partial\tree$,
the map $\emap(\gamma;~\cdot~):\left(0,\infty\right)\rightarrow\R$
is Lipschitz continuous.
\end{lem}

\begin{proof}
Fix $\gamma\in\partial\tree$. To prove the Lipschitz continuity (in
$V$) of $\emap(\gamma;V):\partial\tree\rightarrow\sigma(\Ham)$,
we first argue that the spectral bands are Lipschitz continuous in
$V$. This follows by standard arguments using the operator norm estimate
$\|H_{\alk,V}-H_{\alk,V'}\|\leq|V-V'|$. Therefore if $I$ is a spectral
band of $\sigk$, then for all $V,V'>0$ 
\begin{equation}
\max\big\{|L\left(I(V)\right)-L\left(I(V')\right)|,|R\left(I(V)\right)-R\left(I(V')\right)|\}\leq|V-V'|.\label{eq: Lipschitz_spectral_bands}
\end{equation}
Denoting $\gamma=\left(u_{0},u_{1},u_{2},\ldots\right)$ we have that
$\left\{ \left(\Psi(u_{m})\right)(V)\right\} _{m\in\N}$ is an infinite
nested sequence of compact intervals such that $\bigcap_{m\in\N}\left(\Psi(u_{m})\right)(V)=\left\{ \emap(\gamma;V)\right\} $.
In particular, this means that $\lim_{m\rightarrow\infty}L\left(\left(\Psi(u_{m})\right)(V)\right)=\emap(\gamma;V)$.
Hence, we get for all $V_{1},V_{2}>0$ that 
\begin{align*}
\left|\emap(\gamma;V_{1})-\emap(\gamma;V_{2})\right| & =\lim_{m\rightarrow\infty}\left|L\left(\left(\Psi(u_{m})\right)(V_{1})\right)-L\left(\left(\Psi(u_{m})\right)(V_{2})\right)\right|\leq\left|V_{1}-V_{2}\right|,
\end{align*}
where the last estimate follows from (\ref{eq: Lipschitz_spectral_bands}).
\end{proof}

\begin{lem}
\label{lem: IDOS is independent of V} [also Theorem \ref{thm: paths and spectra}~(\ref{enu: thm-paths and spectra - IDOS}),~(\ref{enu: thm-paths and spectra - negative V})]
Let $\alpha\in[0,1]\setminus\Q$. Then there exists a function $N_{\alpha}:\partial\tree\rightarrow[0,1]$
such that for all $\gamma\in\partial\tree$ and all $V>0$,
\[
\IDS\left(\emap(\gamma;V)\right)=N_{\alpha}(\gamma).
\]
In addition, we have the following spectral properties connecting
between negative and positive values of $V$:
\end{lem}

\begin{enumerate}
\item \label{enu: lem: IDOS is independent of V - 2} For all $V\in\R$,
\[
\sigma(\Ham)=-\sigma(H_{\alpha,-V}).
\]
\item \label{enu: lem: IDOS is independent of V - 3} For all $\gamma\in\partial\tree$
and all $V<0$,
\[
N_{\alpha,V}\left(-\emap(\gamma;-V)\right)=1-N_{\alpha}(\gamma).
\]
\end{enumerate}
\begin{proof}
We start by proving the existence of the function $N_{\alpha}$. In
order to prove this, we argue that the value of the IDS $N_{\alpha,V}(E)$
is related to the spectral bands of the associated periodic approximations.
In the following we consider $\left\{ E\right\} =[E,E]$ as an interval
and so we can use the notation $I\prec\left\{ E\right\} $ if $I$
is another interval. First, observe that we can evaluate the limit
in the definition of the IDS, (\ref{eq: definition of DOS}), by 
\begin{equation}
\IDS(E)=\lim_{k\to\infty}\frac{\#\set{I~\textrm{is a spectral band of}~\sigma(\Halk)}{I\prec\left\{ E\right\} }}{q_{k}}.\label{eq: IDS_by_counting_bands}
\end{equation}
That is, instead of counting eigenvalues of the truncated operator,
we can count spectral bands of the corresponding periodic operator,
see \cite[Prop.~5.2]{BaBeBiTh22}.  The argument justifying this is
based on standard Floquet-Bloch theory combined with finite rank perturbations
of the operators (which may change the numerator above by a bounded
value, not affecting the limit). Let $\gamma\in\partial\tree$ and
$V>0$. Denote $\gamma=\left(u_{0},u_{1},u_{2},\ldots\right)$, and
for each $j\geq0$ denote by $k(u_{j})$ the level of the vertex $u_{j}$.
For example, $k(u_{0})=-1$, since $u_{0}$ is the root of $\tree$.
Note that $\left\{ k(u_{j})\right\} _{j\geq0}$ is an increasing sequence
and that $1\leq k(u_{j+1})-k(u_{j})\leq2$. Specializing (\ref{eq: IDS_by_counting_bands})
for $E=\emap(\gamma;V)$, we get

\begin{align}
\IDS(\emap(\gamma;V)) & =\lim_{k\to\infty}\frac{\#\set{I~\textrm{is a spectral band of}~\sigma(\Halk)}{I\prec\left\{ \emap(\gamma;V)\right\} }}{q_{k}}\label{eq: lem - IDOS is independent of V - IDOS calculation}\\
 & =\lim_{j\to\infty}\frac{\#\set{I~\textrm{is a spectral band of}~\sigma(H_{\alpha_{k(u_{j})},V})}{I\prec\Psi(u_{j})(V)}}{q_{k(u_{j})}}\nonumber \\
\quad & =\lim_{j\to\infty}\frac{\#\set{w~\textrm{is in level }k(u_{j})}{w\prec u_{j}}}{q_{k(u_{j})}}=:N_{\alpha}(\gamma),\nonumber 
\end{align}
where, moving to the second line we use that the limit exists and
henceforth may take a subsequence, $\left\{ k(u_{j})\right\} _{j\in\N}$
of $\left\{ k\right\} _{k\in\N}$. In addition, we used $\emap(\gamma;V)\in\left(\Psi(u_{j})\right)(V)$
for $j\in\N$, which holds by construction of $E_{\alpha}$. Moving
to the last line uses the order-preserving bijection between spectral
bands and the tree vertices (Theorem~\ref{thm: vertices are spectral bands}~(\ref{enu: thm-vertices are spectral bands - left-right relation})).
Finally, noting that the last line clearly is independent of $V$
and depends only on the embedding of the path $\gamma$ within $\tree$,
we may denote it by $N_{\alpha}(\gamma)$.

We continue with treating the case $V<0$ and proving the second part
of the lemma. By \cite[Lem. 4.1]{BaBeLo_spec26} we have that for
all $V\in\R$ and all $k\in\Nz$, $\sigma(\Halk)=-\sigma(H_{\alk,-V})$.
By Proposition~\ref{prop: Monotonicity and limit of spectral approximants}
we have
\[
\sigma(\Ham)=\lim_{k\rightarrow\infty}\left(\sigma(\Halk)\cup\sigma(H_{\alpha_{k+1},V})\right).
\]
Combining this with $\sigma(\Halk)=-\sigma(H_{\alk,-V})$ yields $\sigma(\Ham)=-\sigma(H_{\alpha,-V})$
and proves property (\ref{enu: lem: IDOS is independent of V - 2}).
In order to prove (\ref{enu: lem: IDOS is independent of V - 3}),
we return to the calculation (\ref{eq: lem - IDOS is independent of V - IDOS calculation})
and write, for all $V<0$,
\begin{align*}
\IDS(E) & =\lim_{k\to\infty}\frac{\#\set{I~\textrm{is a spectral band of}~\sigma(\Halk)}{I\prec\left\{ E\right\} }}{q_{k}}\\
 & =\lim_{k\to\infty}\frac{\#\set{I~\textrm{is a spectral band of}~\sigma(H_{\alk,-V})}{I\succ\left\{ -E\right\} }}{q_{k}}\\
 & =\lim_{k\to\infty}\frac{q_{k}-\#\set{I~\textrm{is a spectral band of}~\sigma(H_{\alk,-V})}{I\prec\left\{ -E\right\} }}{q_{k}}\\
 & =1-\lim_{k\to\infty}\frac{\#\set{I~\textrm{is a spectral band of}~\sigma(H_{\alk,-V})}{I\prec\left\{ -E\right\} }}{q_{k}}=1-N_{\alpha,-V}(-E),
\end{align*}
where moving to the second line is justified by $\sigma(\Halk)=-\sigma(H_{\alk,-V})$
and in moving to the third line the numerators might differ by at
most one, but this does not affect the limit. Now, for $V<0$ and
$\gamma\in\partial\tree$, we substitute above $E=-E_{\alpha}(\gamma;-V)$
and so Equation~(\ref{eq: lem - IDOS is independent of V - IDOS calculation})
implies statement (\ref{enu: lem: IDOS is independent of V - 3}).
\end{proof}

\begin{rem}
\label{rem: tree construction V<0}We note that the tree $\tree$
and the map $E_{\alpha}$ were constructed to encode the spectral
information for positive values of $V$ (with the exception of Theorem~\ref{thm: paths and spectra}
(\ref{enu: thm-paths and spectra - negative V})). In light of Lemma~\ref{lem: IDOS is independent of V}
one may wonder whether it is possible to provide an analogous tree
graph to reflect the spectral properties of $\sigma(\Halk)$ and $\sigma(\Ham)$
for $V<0$. Indeed, Lemma~\ref{lem: IDOS is independent of V} may
be used to show that such a tree will be a reflection of the original
tree graph, $\tree$. Returning to the construction of $\tree$ as
described in Section~\ref{subsec: Spectral Approximants Tree}, one
may construct the analogous $V<0$ tree by just replacing the order
of the two vertices $a^{0},b^{0}$ which are connected to the root,
such that $b^{0}\prec a^{0}$. Obviously, the bijection $\Psi$ in
Theorem~\ref{thm: vertices are spectral bands} should be reflected
accordingly.
\end{rem}

The representation of the the IDS $\IDS$ via $\partial\tree$ allows
to provide an explicit expression of the IDS. This is not needed for
the proofs of our results, so we defer the presentation of the expression
and the exact details to Appendix~\ref{App: IDS-explicit-formula}.
\begin{lem}
\label{lem: map boundary-spectrum is surjective} [the surjectivity part of Theorem \ref{thm: paths and spectra}~(\ref{enu: thm-paths and spectra - bijection})]
Let $\alpha\in[0,1]\setminus\Q$ and $V>0$. The map $\emap(~\cdot~;V):\partial\tree\rightarrow\sigma(\Ham)$
is surjective.
\end{lem}

\begin{proof}
Let $V>0$ and $E\in\sigma(\Ham)$. By Proposition~\ref{prop: Monotonicity and limit of spectral approximants},
we have $E\in\sigk(V)\cup\sigkp(V)$ for all $k\in\N$ and so $E$
is contained in a spectral band of $\sigk(V)$ or $\sigkp(V)$. By
Theorem~\ref{thm: vertices are spectral bands}, there is a bijection
$\Psi$ between spectral bands and vertices of the spectral $\alpha$-tree
$\tree$. We conclude that we can choose a sequence of vertices, $\left\{ w_{j}\right\} _{j\in\N}$
in $\tree$, such that $E\in\Psi(w_{j})(V)$. In particular this sequence
of vertices may be chosen such that for all $j\in\N$, $w_{j}$ is
in level $k_{j}$ and $k_{j}<k_{j+1}$.

Since $\tree$ is a directed connected tree, we can choose for each
$j\in\N$ an infinite path $\gamma_{j}\in\partial\tree$ passing through
the vertex $w_{j}$, see a sketch in Figure~\ref{Fig: order preserving}~(1).
We note that $\partial\tree$ is a compact space (using the product
topology on spheres from the root, or equivalently noticing that this
space is a Gromov boundary). Thus, $\left\{ \gamma_{j}\right\} _{j\in\N}$
admits a convergent subsequence $\left\{ \gamma_{j_{l}}\right\} _{l\in\N}$
with limit $\gamma=(u_{0},u_{1},u_{2},\ldots)\in\partial\tree$. We
claim that $\emap(\gamma~;V)=E$, which proves surjectivity.

By the convergence of $\left\{ \gamma_{j_{l}}\right\} _{l\in\N}$
to $\gamma$, there is for each $N\in\N$, an $l_{N}\in\N$ such that
for all $l\geq l_{N}$, we have $\gamma|_{[0,N]}=\gamma_{j_{l}}|_{[0,N]}$
and $k_{j_{l}}>N$ (i.e. the vertex $w_{j_{l}}$ does not appear on
the first $N$ vertices of the path $\gamma_{j_{l}}$). Thus, for
each $N\in\N$ and $l\geq l_{N}$, we have $u_{N}\to w_{j_{l}}$ (a
path from $u_{N}$ to the vertex $w_{j_{l}}$). Hence, Theorem~\ref{thm: vertices are spectral bands}~(\ref{enu: thm-vertices are spectral bands - inclusion})
implies $\Psi(w_{j_{l}})(V)\strict\Psi(u_{N})(V)$. Since by construction
$E\in\Psi(w_{j_{l}})(V)$, we conclude $E\in\Psi(u_{N})(V)$ for all
$N\in\N$. Hence, $E=\emap(\gamma~;V)$ follows from the definition
of $\emap$ by the intersection of all $\Psi(u_{N})(V)$.
\end{proof}

Next, we show that the map $\emap(~\cdot~;V):\partial\tree\rightarrow\sigma(\Ham)$
is order preserving. Towards this, recall the order relation on $\partial\tree$.
Let $\gamma_{1}=\left(u_{0},u_{1},\ldots\right),\gamma_{2}=\left(w_{0},w_{1},\ldots\right)\in\partial\tree$.
If $\gamma_{1}=\gamma_{2}$, we set both $\gamma_{1}\preceq\gamma_{2}$
and $\gamma_{2}\preceq\gamma_{1}$ (so that the order is reflexive).
Otherwise, there exists a unique $k\in\N_{0}$ such that $u_{k-1}=w_{k-1}$
and $u_{k}\neq w_{k}$. By Definition~\ref{def: spectral approximants tree},
either
\begin{itemize}
\item $u_{l}\prec w_{l}$ for all $l\geq k$ and so $\gamma_{1}\preceq\gamma_{2}$,
or
\item $w_{l}\prec u_{l}$ for all $l\geq k$ and so $\gamma_{2}\preceq\gamma_{1}$.
\end{itemize}
\begin{figure}[hbt]
\includegraphics[scale=0.86]{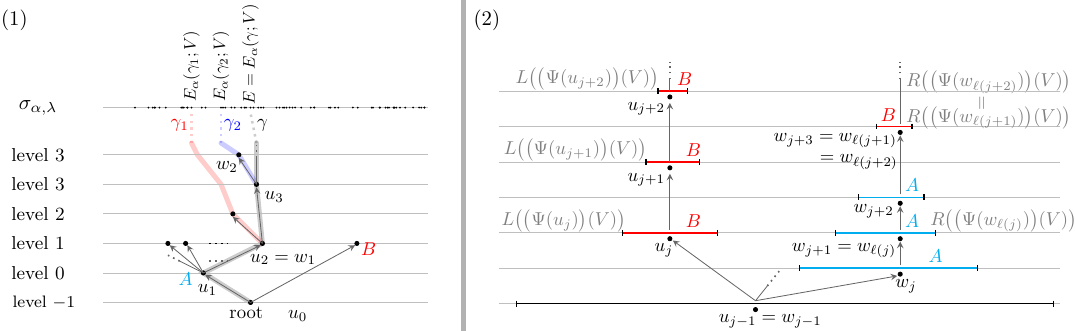}\caption{In (1), the sequence of paths and vertices constructed in Lemma~\ref{lem: map boundary-spectrum is surjective}
are outlined. In (2), the paths $\gamma_{1}$ and $\gamma_{2}$ and
their associated spectral bands in the proof of Lemma~\ref{lem: order preserving E_=00005Calpha}
are sketched.}
\label{Fig: order preserving}
\end{figure}

\begin{lem}
\label{lem: order preserving E_=00005Calpha} [also Theorem \ref{thm: paths and spectra}~(\ref{enu: thm-paths and spectra - order preserving})]
Let $\alpha\in[0,1]\setminus\Q$. If $\gamma_{1},\gamma_{2}\in\partial\tree$
satisfy $\gamma_{1}\preceq\gamma_{2}$, then $E_{\alpha}(\gamma_{1};V)\leq E_{\alpha}(\gamma_{2};V)$
for all $V>0$.
\end{lem}

\begin{proof}
Let $V>0$ and $\gamma_{1}=\left(u_{0},u_{1},\ldots\right),\gamma_{2}=\left(w_{0},w_{1},\ldots\right)\in\partial\tree$
be such that $\gamma_{1}\preceq\gamma_{2}$. If $\gamma_{1}=\gamma_{2}$
then $E_{\alpha}(\gamma_{1};V)=E_{\alpha}(\gamma_{2};V)$ by definition,
and we may proceed assuming $\gamma_{1}\neq\gamma_{2}$. Let $k\in\N_{0}$
be such that $u_{i}\prec w_{j}$ for all $i,j\geq k$ . It is worth
pointing out that neither the vertex $u_{j}$ (respectively $w_{j}$)
is necessarily in level $j$ nor both $u_{j},w_{j}$ are in the same
level. 

However, since two vertices connected by an edge differ at most by
two levels, we conclude that we can choose a map $\ell:\N\to\N$ with
$\lim_{j\to\infty}\ell(j)=\infty$ such that $u_{j}$ and $w_{\ell(j)}$
are at most one level apart (so, they are either in the same level
or in consecutive levels). Such a map is in general not unique - an
example is depicted in Figure~\ref{Fig: order preserving}~(2). 

For $k$ as above, there is a $j_{0}\in\N$ such that for $j\geq j_{0}$,
we have $j\geq k$ and $\ell(j)\geq k$. Thus, $u_{j}\prec w_{\ell(j)}$
holds for $j\geq j_{0}$. Then Theorem~\ref{thm: vertices are spectral bands}~(\ref{enu: thm-vertices are spectral bands - left-right relation})
implies $\left(\Psi(u_{j})\right)(V)\prec\left(\Psi(w_{\ell(j)})\right)(V)$
for $j\geq j_{0}$ using that $u_{j}$ and $w_{\ell(j)}$ are at most
one level apart. Hence, $L\left(\left(\Psi(u_{j})\right)(V)\right)<R\left(\left(\Psi(w_{\ell(j)})\right)(V)\right)$
for $j\geq j_{0}$ follows from Definition~\ref{def: order relations on spectral bands - as maps}.
By construction of $E_{\alpha}$, we have 
\[
E_{\alpha}(\gamma_{1};V)=\lim_{j\to\infty}L\left(\left(\Psi(u_{j})\right)(V)\right)\quad\textrm{and}\quad E_{\alpha}(\gamma_{2};V)=\lim_{j\to\infty}R\left(\left(\Psi(w_{\ell(j)})\right)(V)\right),
\]
and so we conclude $E_{\alpha}(\gamma_{1};V)\leq E_{\alpha}(\gamma_{2};V)$.
\end{proof}

We note that the statement of Lemma~\ref{lem: order preserving E_=00005Calpha}
can be strengthened: If $\gamma_{1}\preceq\gamma_{2}$ and $\gamma_{1}\neq\gamma_{2}$,
then $E_{\alpha}(\gamma_{1};V)<E_{\alpha}(\gamma_{2};V)$ (not equal!)
follows for all $V>0$. This is an immediate consequence of the injectivity
of $\emap(~\cdot~;V):\partial\tree\rightarrow\sigma(\Ham)$, which
is proven in Section~\ref{sec: Injectivity-Proof}.

\section{Tools for proving the injectivity of $\protect\emap$\label{sec: Tools for Injectivity proof}}

In order to complete the proof of Theorem~\ref{thm: paths and spectra},
it remains to establish the injectivity of the map $\emap(~\cdot~;V):\partial\tree\rightarrow\sigma(\Ham)$.
This requires a change of perspective together with some additional
techniques. Both are developed in the present section, while the actual
injectivity proof is postponed to Section~\ref{sec: Injectivity-Proof}.

The change of viewpoint is introduced in Subsection~\ref{subsec: Finite Continued Fraction}
through the space of finite continued fraction expansions $\Co$.
In Subsection~\ref{subsec: spectra via traces}, this $\Co$-space
is connected to the classical formalism of transfer matrices and traces.
Next, Subsection~\ref{subsec: A-B-types-via-C-space} uses the $\Co$-space
to further develop the hierarchical structure of spectral bands.

In Subsection~\ref{subsec: admissibility-and-triple-trace-product},
we prove positivity properties of certain trace products using a new
approach for the admissibility concept (which we introduced in \cite{BaBeLo_spec26}).
When $V$ tends to zero, spectral bands start to overlap and trace
values may change sign. Therefore, in Subsection~\ref{subsec: trace-monotonicity},
we use the tools developed throughout this section to show that these
signs remain unchanged.

For a first reading, one may read only the first two subsections and
then proceed directly to Section~\ref{sec: Injectivity-Proof}. The
later subsections are needed for the proof of Lemma~\ref{lem: estimating t_=00007Bc.m.n=00007D},
which may initially be viewed as a black box in order to focus on
the key ideas in the proof the injectivity of $E_{\alpha}$ (Lemma~\ref{lem: map boundary-spectrum is injective}).

\subsection{The space $\protect\Co$ of augmented finite continued fraction expansions\label{subsec: Finite Continued Fraction}}

We have already recognized the importance of the rational approximations,
$\Halk$ and their spectra $\sigma_{k}(V)$. We also observed the
significant role played by the continued fraction expansion of $\alk$,
\[
\alk=c_{0}+\frac{1}{c_{1}+\frac{1}{\ddots+\frac{1}{c_{k}}}}\in\left[0,1\right].
\]
We write the sequence of coefficients above as a tuple $\co(k):=[0,c_{0},c_{1},\ldots,c_{k}]$.
Writing it in this form implies $c_{-1}=c_{0}=0$. That $c_{0}=0$
is clear from $\alk\in[0,1]$, whereas the additional entry $c_{-1}=0$
corresponds to level $k=-1$ (i.e., the root vertex) of the spectral
$\alpha$-tree $\tree$. Next, we change our perspective. Rather than
focusing on a single sequence of rational approximations, $\left\{ \alk\right\} $,
of a particular $\alpha\notin\Q$, we consider the whole space of
augmented finite continued fraction expansions, and draw connections
between its elements, as defined next.

Define the space of \emph{augmented finite continued fraction expansions}
by 
\[
\Co:=\left\{ [0],~[0,0]\right\} \cup\bigcup_{k\in\N}\set{[0,0,c_{1},\ldots,c_{k}]}{c_{1},\ldots,c_{k-1}\in\N,~c_{k}\in\N_{-1}},
\]
where we are using the convention that the two first entries of all
$\co\in\Co$, satisfy $c_{-1}=c_{0}=0$. A special emphasis is given
to the non-standard choice of allowing $c_{k}$ attain the values
$0$ and $-1$. Additionally, denote 
\[
[\co,m]:=[0,0,c_{1},\ldots,c_{k},m],\qquad m\in\N_{-1}.
\]
This notation is used only when $[\co,m]\in\Co$. To assure this,
we assume when using the notation $[\co,m]$ that either $\co=[0,0]$
or $\co=[0,0,c_{1},\ldots,c_{k}]\in\Co$ with $c_{k}\not\in\{-1,0\}$.

The connection between the finite continued fraction expansions and
the rational numbers is done via the \emph{evaluation map} $\varphi:\Co\to\R\cup\left\{ \infty\right\} $.
This map is defined for all $\co=[0,c_{0},c_{1},\dots,c_{k}]\in\Co\backslash\left\{ [0]\right\} $
by
\begin{equation}
\varphi([0,c_{0},c_{1},\dots,c_{k}]):=\begin{cases}
\varphi([0,c_{0},c_{1},\dots,c_{k-2}]), & k\in\N\textrm{ and }c_{k}=0,\\[0.1cm]
c_{0}+\frac{1}{c_{1}+\frac{1}{\ddots+\frac{1}{c_{k}}}}, & \text{otherwise},
\end{cases}\label{eq: phi map from C to Q}
\end{equation}
and $\varphi([0]):=\infty$. Note that $\varphi([0,c_{0},c_{1},\dots,c_{k},-1])=\varphi([0,c_{0},c_{1},\dots,c_{k}-1])$
for $k\geq1$.

In all statements and proofs so far, we fixed some $\alpha\in\left[0,1\right]\backslash\Q$
and considered all its finite continued fraction expansions, giving
rise to $\alk=\frac{p_{k}}{q_{k}}$. From this point and later on
we consider the space of all rational numbers represented by their
augmented finite continued fraction expansions, $\co\in\Co$.

\subsection{The spectra $\left\{ \protect\sigc\right\} _{\protect\co\in\protect\Co}$
via traces of transfer matrices\label{subsec: spectra via traces}}

We present the well-known formalism for transfer matrices, though
adapted to the $\Co$-space introduced above. This describes the rational
approximants spectra $\sigma_{k}(V)$. For $V\in\mathbb{R}$, define
\[
M_{[0]}(E,V):=\begin{pmatrix}1 & -V\\
0 & 1
\end{pmatrix},\qquad M_{[0,0]}(E,V):=\begin{pmatrix}E & -1\\
1 & 0
\end{pmatrix}
\]
and recursively define the transfer matrices for $\co=[0,0,c_{1},\ldots,c_{k}]\in\Co$
(where $k\in\N$) by
\[
M_{\co}(E,V):=M_{[0,0,c_{1},\ldots,c_{k-2}]}(E,V)M_{[0,0,c_{1},\ldots,c_{k-1}]}(E,V)^{c_{k}}.
\]
Consequently, denote the traces of the transfer matrices by
\begin{equation}
\tc(E,V):=\tr(M_{\co}(E,V)).\label{eq: traces}
\end{equation}

Representing the spectra via the traces of these transfer matrices
is a classical approach \cite{Casdagli1986,Sut87,BIST89,Raym95}.
Our description only slightly deviates from the conventional one,
by referring to all the elements of $\Co$ (within the literature
above we take a route which is the closest to \cite{Raym95}). This
approach is expressed in the next definition and proposition.
\begin{defn}
\label{def: sigma_c}For all $V\in\R$, and $\co\in\Co$ denote 
\[
\sigc(V):=\set{E\in\R}{\left|\tc(E,V)\right|\leq2}.
\]
\end{defn}

The connection between the evaluation map $\varphi$, the $\Co$-space,
the traces and the spectra is as follows.
\begin{lem}
\label{lem: Equal-co-values}\cite[Prop.~3.5, Lem.~3.6]{BaBeBiTh22}For
all $\co,\widetilde{\co}\in\Co$ with $\varphi(\widetilde{\co})=\varphi(\co)$,
we have
\[
\sigma_{\cop}(V)=\sigc(V)\quad\textrm{and}\quad t_{\widetilde{\co}}(E,V)=\tc(E,V),\quad\textrm{for all}\;E,V\in\R.
\]
Let $\alpha\in[0,1]\setminus\Q$ with infinite continued fraction
expansion $(0,c_{1},c_{2},\ldots)$ and $k\in\N$. Then the spectrum
$\sigma_{k}(V)$ of the operator $H_{\alk,V}$ satisfies
\[
\sigma_{\co(k)}(V)=\sigma_{k}(V),\qquad\text{where }\co(k):=[0,c_{0},c_{1},\ldots,c_{k}].
\]
\end{lem}

We collect in the following proposition some well-known identities
of the traces, see e.g. \cite{Raym95,BIST89,Simon2011,BaBeBiTh22,DaFi24-book_2}.
Recall that for $\co\in\Co$, $\tc$ is a function of $E,V\in\R$,
but we abbreviate notation and suppress these dependencies in the
following.
\begin{prop}
\label{prop: traceMaps} Let $\co\in\Co$ and $m\in\Nz$ such that
$[\co,m-1]\in\Co$. Then the following holds.
\begin{enumerate}
\item \label{enu:.Prop-traceMaps-Fricke-Vogt}We have for all $V\in\R$
\textup{(the }\textup{Fricke}\textup{--}\textup{Vogt}\textup{ invariant)}
\[
V^{2}+4=t^{2}_{\co}+t^{2}_{[\co,m]}+t^{2}_{[\co,m-1]}-\tc t_{[\co,m]}t_{[\co,m-1]}
\]
\item \label{enu:.Prop-traceMaps-Recursive Relation}We have
\[
t_{[\co,m+1]}=\tc t_{[\co,m]}-t_{[\co,m-1]}.
\]
\item \label{enu:.Prop-traceMaps-3 intersection property}For $V>4$, we
have
\[
\sigc(V)\cap\sigcm(V)\cap\sigma_{[\co,m-1]}(V)=\emptyset.
\]
\item \label{enu:.Prop-traceMaps-band edges}For $E\in\R$ and $\co\in\Co$
with $\varphi(\co)\in[0,1]$, we have $\left|t_{\co}(E,V)\right|=2$,
if and only if $E\in\left\{ L\left(\Ic(V)\right),R\left(\Ic(V)\right)\right\} $
for some spectral band $\Ic(V)$ in $\sigc(V)$.
\item \label{enu:.Prop-traceMaps-trace-monotone} For a spectral band $\Ic(V)$
in $\sigc(V)$ the map $\left.\tc\right|_{\Ic}:\Ic\to[-2,2]$ is strictly
monotone, continuous and onto.
\end{enumerate}
\end{prop}

\begin{proof}
The statements above appear in: (\ref{enu:.Prop-traceMaps-Fricke-Vogt})
\cite[Prop.~3.13]{BaBeBiTh22}, (\ref{enu:.Prop-traceMaps-Recursive Relation})
\cite[Lem.~3.8]{BaBeBiTh22}, (\ref{enu:.Prop-traceMaps-3 intersection property})
\cite[Prop.~4.7]{BaBeBiTh22}, (\ref{enu:.Prop-traceMaps-band edges})
\cite[Prop.~3.5, Prop.~4.1]{BaBeBiTh22} and (\ref{enu:.Prop-traceMaps-trace-monotone})
\cite[Prop.~3.5]{BaBeBiTh22}.
\end{proof}

\subsection{The tower property. \label{subsec: A-B-types-via-C-space}}

We present here an additional hierarchical property of the spectral
bands of $\sigc$. This property is related to the spectral band structure
stated in Proposition~\ref{prop: A-B-types-from-algebra-paper} and
we call it the tower property (see also \cite[Def.~2.9]{BaBeLo_spec26}).
To describe it we first fix $\alpha\in[0,1]\backslash\Q$. and its
continued fraction expansion $(0,c_{1},c_{2},\ldots)$, but then vary
one digit in this expansion.

In the following, we study the spectral bands $I:V\mapsto I(V)$,
$V>0$, as maps in the sense of Definition~\ref{def: A spectral band is continuous}.
By Proposition~\ref{prop: A-B-types-from-algebra-paper}, their type
($\tA$ or $\tB$) is independent of $V>0$. Let $I_{k}$ be a spectral
band of $\sigma_{k}$ corresponding to $\co(k)=[0,0,c_{1},\ldots,c_{k}]$.
For $n\in\N$, consider the spectra $\sigma_{[\co(k+1),n]}$. By Proposition~\ref{prop: A-B-types-from-algebra-paper}\,(\ref{enu: forward-B}),
there are exactly $M+1$ spectral bands of type $\tB$, denoted here
by $B^{1}_{n},\ldots,B^{M+1}_{n}$ which are strictly contained in
$I_{k}$. We consider different values of $n\in\N$ (this is the expansion
digit which we vary) and note that $\varphi([\co(k+1),n])=\alpha_{k+2}$
if and only if $n=c_{k+2}$ (and otherwise $\varphi([\co(k+1),n])$
deviates from the standard rational approximations of $\alpha$).
The following proposition shows that the bands $B^{1}_{n},\ldots,B^{M+1}_{n}$
satisfy nested inclusion relations for varying $n$, which we call
the tower property.

\begin{figure}[hbt]
\includegraphics[scale=0.82]{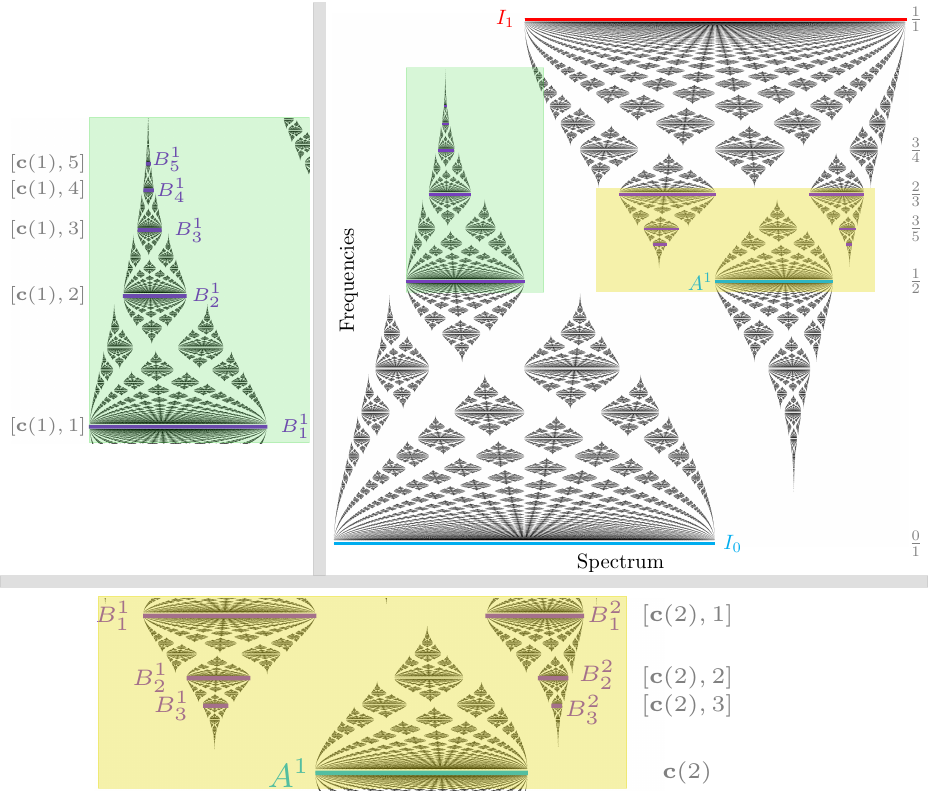}\caption{Plot of the Kohmoto butterfly. The highlighted spectral bands illustrate
the tower property from Proposition~\ref{prop: Tower Property} bands
form nested towers as the last continued-fraction digit is increased.
In particular, the figure shows the $\protect\tB$-towers over $I_{0}=[-2,2]$
(for $\protect\co(1)=[0,0,1]$) and $I_{1}=[-2+V,2+V]$ (for $\protect\co(2)=[0,0,1,1]$),
respectively. Zoomed-in insets are included.\label{Fig: TreeData-1}}
\end{figure}

\begin{prop}
\label{prop: Tower Property} Let $\alpha\in[0,1]\setminus\Q$ with
infinite continued fraction expansion $(0,c_{1},c_{2},\ldots)$ and
$k\in\N$. Let $I_{k}$ be a spectral band of $\sigma_{k}$. Denote
\[
M:=\begin{cases}
c_{k+1}-1 & I_{k}~\textrm{is of type }\tA\\
c_{k+1} & I_{k}~\textrm{is of type }\tB
\end{cases}.
\]
For each $n\in\N$, let $\left\{ B^{j}_{n}\right\} ^{M+1}_{j=1}$
be the unique type $\tB$ spectral bands in $\sigma_{[\co(k+1),n]}$
which are contained in $I_{k}$ as in Proposition~\ref{prop: A-B-types-from-algebra-paper}.
Then for all $1\leq j\leq M+1$
\[
B^{j}_{n+1}\strict B^{j}_{n}\strict I_{k}.
\]
\end{prop}

\begin{proof}
Let $n\in\N$ and $1\leq j\leq M+1$ be fixed. By Proposition~\ref{prop: A-B-types-from-algebra-paper}\,(\ref{enu: forward-B})
$B^{j}_{n}\strict I_{k}$, so we need to prove just that $B^{j}_{n+1}\strict B^{j}_{n}$.
By \cite[Prop.~3.4]{BaBeLo_spec26}, the spectral band $B^{j}_{n+1}$
in $\sigma_{[\co(k+1),n+1]}$ is of type $\tB$ if and only if there
is a spectral band $I^{j}$ in $\sigma_{[\co(k+1),n+1,-1]}=\sigma_{[\co(k+1),n]}$
such that $B^{j}_{n+1}\strict I^{j}$. We show that $I^{j}=B^{j}_{n}$. 

First $I^{j}$ must be of type $\tB$. Otherwise $I^{j}\strict\sigma_{\co(k+1)}=\sigma_{k+1}$,
which implies $B^{j}_{n+1}\strict\sigma_{k+1}$. This contradicts
that $B^{j}_{n+1}$ is of type $\tB$. 

Thus, $I^{j}$ is of type $\tB$ and hence strictly contained in $\sigma_{k}$.
Since $B^{j}_{n+1}\strict I^{j}$ and $B^{j}_{n+1}\strict I_{k}$,
we have $I^{j}(V)\cap I_{k}(V)\neq\emptyset$ for all $V>0$. But
this means that $I_{k}$ must be the spectral band of $\sigk$ which
contains $I^{j}$, i.e., $I^{j}\strict I_{k}$. By Proposition~\ref{prop: A-B-types-from-algebra-paper}\,(\ref{enu: forward-B})
there are exactly $M+1$ spectral bands of $\sigma_{[\co(k+1),n]}$
of type $\tB$ which are strictly contained in $I_{k}$, and so $I^{j}$
must be one of them, i.e. one of $\left\{ B^{j}_{n}\right\} ^{M+1}_{j=1}$.
The interlacing property in Proposition~\ref{prop: A-B-types-from-algebra-paper}\,(\ref{enu: interlacing}),
together with $B^{j}_{n+1}\strict I^{j}$ guarantees that $I^{j}=B^{j}_{n}$.
\end{proof}

\subsection{Admissibility and triple trace products\label{subsec: admissibility-and-triple-trace-product}}

We need to relate the hierarchical structure in Proposition~\ref{prop: A-B-types-from-algebra-paper}
and Proposition~\ref{prop: Tower Property} to various trace values
$\tc$. This is mainly done in the next subsection, but towards this
we establish in Lemmas~\ref{lem: Product traces - consecutive B-bands}
and \ref{lem: Product traces - along the tower} the positivity of
certain trace products. From here on we abbreviate the notation for
the trace functions $\tc(E,V)$ and write $\tc(E(V))$ or $\tc(E)$,
whenever $E:(0,\infty)\rightarrow\R$ is taken to be a $V$-dependent
map.

Let $\co=[0,0,c_{1},\ldots,c_{k},c_{k+1}]\in\Co$ such that $[\co,1]\in\Co$.
Let $I_{k}:V\mapsto I_{k}(V)$ be a spectral band in $\sigma_{\co(k)}$
and $E_{k}:V\mapsto E_{k}(V)$ be an edge of $I_{k}$, i.e., either
$E_{k}(V)=L(I_{k}(V))$ or $E_{k}(V)=R(I_{k}(V))$ for all $V>0$.
Similarly, let $E_{k+1}(V)$ (respectively $E_{k+2}(V)$) be an edge
of a spectral band of $\sigma_{\co(k+1)}(V)$ (respectively $\sigma_{[\co(k+1),1]}(V)$).
Following \cite[Prop. 4.8]{BaBeLo_spec26} we say that $E_{k}$, $E_{k+1}$,
$E_{k+2}$ is an admissible triple if the product of corresponding
traces is positive, i.e. if
\[
t_{\co(k)}(E_{k})\thinspace t_{\co(k+1)}(E_{k+1})\thinspace t_{[\co(k+1),1]}(E_{k+2})>0.
\]

The specific continued fractions, $\co(k)$, $\co(k+1)$, $[\co(k+1),1]$
in the definition of admissibility above are chosen so that they correspond
to those which appear in the Fricke--Vogt invariant (Proposition~\ref{prop: traceMaps}\,(\ref{enu:.Prop-traceMaps-Fricke-Vogt})).
Indeed, admissibility is applied in this paper in conjunction with
the Fricke--Vogt invariant. A somewhat more general definition of
admissibility is introduced and applied in \cite{BaBeLo_spec26}.
There, we also give an equivalent criterion for admissibility in terms
of the position of the spectral band in the spectrum. We provide here
a specialized version of this criterion which can be applied in the
current paper. The criterion is presented in the next definition and
lemma and it is then applied in the proofs of the two preceding lemmas.
\begin{defn}
\label{def: index of spectral band}[Index of a spectral band] Let
$I_{k}$ be a spectral band of $\sigma_{k}$. The \emph{index} of
$I_{k}$ (in $\sigma_{k}$) is defined by 
\[
\ind(I_{k}):=\left|\set{I\textrm{ is a spectral band of }\sigma_{k}}{I\prec I_{k}}\right|.
\]
\end{defn}

Note that the index counting starts from zero, namely $0\leq\ind(I_{k})\leq q_{k}-1$
where $\alpha_{k}=\frac{p_{k}}{q_{k}}$ with $p_{k},q_{k}$ coprime.
Moreover, we emphasize that $\ind(I_{k})$ is independent of $V>0$.
\begin{lem}
\label{lem: admissibility-criterion}[Admissibility criterion] Let
$\co=[0,0,c_{1},\ldots,c_{k},c_{k+1}]\in\Co$ such that $[\co,1]\in\Co$.
Let $E_{k}$, $E_{k+1}$, $E_{k+2}$ be edges of spectral bands $I_{k}$,
$I_{k+1}$, $I_{k+2}$ of $\sigma_{\co(k)}(V)$, $\sigma_{\co(k+1)}(V)$,
$\sigma_{[\co(k+1),1]}(V)$, correspondingly. Further assume that
two out of $E_{k}(V)$, $E_{k+1}(V)$, $E_{k+2}(V)$ are left edges
of spectral bands and one of them is a right edge of a spectral band.
Then $E_{k}$, $E_{k+1}$, $E_{k+2}$ are admissible if and only if
\begin{equation}
\ind(I_{k})+\ind(I_{k+1})+\ind(I_{k+2})\equiv1\mod 2.\label{eq: admissibility-criterion}
\end{equation}
\end{lem}

\begin{proof}
The lemma may be obtained as a corollary of \cite[Lem. 5.6]{BaBeLo_spec26}.
Alternatively, one may prove this directly. To do so, note first that
if $\sigma_{\tilde{\co}}$ has $q$ spectral bands then the sign of
$t_{\tilde{\co}}$ at an edge of a spectral band $I$ satisfies
\[
\sgn\left(t_{\tilde{\co}}(L(I))\right)=(-1)^{q-\ind(I)}\quad\textrm{and}\quad\sgn\left(t_{\tilde{\co}}(R(I))\right)=(-1)^{q-\ind(I)+1},
\]
as can be inferred for example from \cite[Lem. 4.6, Lem. 5.4]{BaBeLo_spec26}.
This, combined with the recursion $q_{[\co,m,1]}=\qcm+\qc$ from (\ref{eq: Proof vertices are spectral bands - recursion q_k})
allows to conclude (\ref{eq: admissibility-criterion}).
\end{proof}

\begin{lem}
\label{lem: Product traces - consecutive B-bands} Let $\alpha\in[0,1]\setminus\Q$
with infinite continued fraction expansion $(0,c_{1},c_{2},\ldots)$
and $k\in\N$. Consider a spectral band $I_{k}$ in $\sigma_{\co(k)}$
with the associated spectral band $B^{1}_{1}$ in $\sigma_{[\co(k+1),1]}$
introduced in Proposition~\ref{prop: A-B-types-from-algebra-paper}.
Moreover, let $J_{k+1}$ in $\sigma_{\co(k+1)}$ be the rightmost
spectral band of $\sigma_{\co(k+1)}$ for which $J_{k+1}\prec I_{k}$.
Then for all $V>0$,
\[
\sgn\left(t_{\co(k)}\left(L(I_{k}(V))\right)\cdot t_{\co(k+1)}\left(R(J_{k+1}(V))\right)\cdot t_{[\co(k+1),1]}\left(L(B^{1}_{1}(V))\right)\right)=+1.
\]
\end{lem}

\begin{proof}
First, we observe that the three spectral bands to which the statement
refers to satisfy the following index relation 
\begin{align}
\ind(B^{1}_{1}) & =\ind(I_{k})+\ind(J_{k+1})+1.\label{eq: index-relation-lem-1}
\end{align}
This can be verified using the spectral bands hierarchical structure
as given in Proposition~\ref{prop: A-B-types-from-algebra-paper},
or alternatively this relation can also be found in \cite[Eq. (5.11)]{BaBeLo_spec26}
(by substituting $\co=\co(k)$, $m=c_{k+1}$, $n=1$ there). With
(\ref{eq: index-relation-lem-1}), the lemma is obtained as a corollary
of Lemma\,\ref{lem: admissibility-criterion}.
\end{proof}

\begin{lem}
\label{lem: Product traces - along the tower} Let $\alpha\in[0,1]\setminus\Q$
with infinite continued fraction expansion $(0,c_{1},c_{2},\ldots)$
and $k,n\in\N$. Consider a spectral band $I_{k}$ in $\sigma_{k}$
with the associated spectral band $B^{1}_{n}$ in $\sigma_{[\co(k+1),n]}$
and $B^{1}_{n+1}$ in $\sigma_{[\co(k+1),n+1]}$ introduced in Proposition~\ref{prop: A-B-types-from-algebra-paper}.
Moreover, let $J_{k+1}$ in $\sigma_{\co(k+1)}$ be the rightmost
spectral band of $\sigma_{\co(k+1)}$ for which $J_{k+1}\prec I_{k}$.
Then for all $V>0$,
\[
\sgn\left(t_{\co(k+1)}\left(R(J_{k+1}(V))\right)\cdot t_{[\co(k+1),n]}\left(L(B^{1}_{n}(V))\right)\cdot t_{[\co(k+1),n+1]}\left(L(B^{1}_{n+1}(V))\right)\right)=+1.
\]
\end{lem}

\begin{proof}
As in the proof of Lemma~\ref{lem: Product traces - consecutive B-bands},
the indices of the relevant spectral bands satisfy a similar relation,
\begin{align}
\ind(B^{1}_{n+1}) & =\ind(B^{1}_{n})+\ind(J_{k+1})+1,\label{eq: index-relation-lem-2}
\end{align}
which can be either verified using the spectral bands hierarchical
structure as given in Proposition~\ref{prop: A-B-types-from-algebra-paper}
and Proposition~\ref{prop: Tower Property} , or alternatively this
relation can also be found in \cite[Eq.~(5.13)]{BaBeLo_spec26} (by
substituting $\co=[\co(k+1),n]$). With (\ref{eq: index-relation-lem-2}),
we may apply Lemma\,\ref{lem: admissibility-criterion} to finish
the proof. To see the correspondence with Lemma~\ref{lem: admissibility-criterion}
we may use the notation $\cop:=[\co(k+1),n]$, and $\tilde{k}:=k+1$
for which $\varphi(\co(k+1))=\varphi(\tilde{\co}(\tilde{k}))$, $\varphi([\co(k+1),n])=\varphi(\tilde{\co}(\tilde{k}+1)))$
and $\varphi([\co(k+1),n+1])=\varphi([\tilde{\co}(\tilde{k}+1),1])$;
the continued fractions at the right hand sides of these relations
correspond to those which appear in Lemma~\ref{lem: admissibility-criterion}.
\end{proof}

\subsection{Spectral band positions and trace identities\label{subsec: trace-monotonicity}}

We develop here useful identities on the relative positions of the
spectral bands and on the corresponding trace values $\tc$ - these
appear in Lemma~\ref{lem: monotonicity of t_=00007Bc.m.n=00007D}.

As in the previous subsection, we continue using here the abbreviated
notation $\tc(E(V))$ or $\tc(E)$ if $E:(0,\infty)\rightarrow\R$
is a $V$-dependent map. Adopting this notation, for any $\co\in\Co$
and $\Ic$ a spectral band of $\sigc$, we have that $\left.\tc\right|_{\Ic}:\Ic\to[-2,2]$
is strictly monotone, continuous and onto (see Proposition~\ref{prop: traceMaps}).
This  implies in particular that $\tc$ vanishes exactly once on $\Ic$.
Explicitly, for each $V\in(0,\infty)$, there is a unique $\Cen(\Ic(V))\in\R$
such that $\tc\left(\Cen\left(\Ic(V)\right)\right)=0$. We consider
$\Cen(\Ic)$ as a map $(0,\infty)\ni V\mapsto\Cen(\Ic(V))$, and call
it the \emph{Zentrum} of the spectral band $\Ic$. Note though that
$\Cen(\Ic)$ is only the center in terms of the image of $\tc$, and
not necessarily equal to $\frac{1}{2}(R(\Ic)-L(\Ic))$. Further observe
that $(0,\infty)\ni V\mapsto\Cen(\Ic(V))$ is continuous by construction.

The next lemma is accompanied by Figure~\ref{Fig: sign trace and central points}.
\begin{lem}
\label{lem: monotonicity of t_=00007Bc.m.n=00007D} Let $\alpha\in[0,1]\setminus\Q$
with infinite continued fraction expansion $(0,c_{1},c_{2},\ldots)$
and $k\in\N$. Consider a spectral band $I_{k}$ in $\sigma_{\co(k)}$
of type $B$. For all $n\in\N$, let $B^{1}_{n}$ be the left-most
type $B$ spectral band in $\sigma_{[\co(k+1),n]}$ which is strictly
contained in $I_{k}$ (as was introduced in Proposition~\ref{prop: A-B-types-from-algebra-paper}\,(\ref{enu: forward-B})).
Furthermore, let $J_{k+1}$ in $\sigma_{\co(k+1)}$ be the rightmost
spectral band of $\sigma_{\co(k+1)}$ for which $J_{k+1}\prec I_{k}$.
If $E(V):=R(J_{k+1}(V))\in I_{k}(V)$ for $V>0$, then the following
statements hold.
\begin{enumerate}
\item \label{enu: lem-monotonicity of t_=00007Bc.m.n=00007D - sign t_c}
We have $E(V)<\Cen(I_{k}(V))$ and $\sgn\left(t_{\co(k)}(E(V))\right)=\sgn\left(t_{\co(k)}\left(L(I_{k}(V))\right)\right)$.
\item \label{enu: lem-monotonicity of t_=00007Bc.m.n=00007D - sign t_=00005Bc.m.n=00005D}
For all $n\in\N$, we have $E(V)<\Cen\left(B^{1}_{n}(V)\right)$ and
\[
\sgn\left(t_{[\co(k+1),n]}(E(V))\right)=\sgn\left(t_{[\co(k+1),n]}\left(L(B^{1}_{n}(V))\right)\right).
\]
\item \label{enu: lem-monotonicity of t_=00007Bc.m.n=00007D - monotonicity t_=00005Bc.m.n=00005D}
If $E(V)\in B^{1}_{1}(V)$, then for all $n\in\Nz$, 
\begin{equation}
\left|t_{[\co(k+1),n+1]}(E(V)))\right|>\left|t_{[\co(k+1),n]}(E(V)))\right|>\ldots>\left|t_{[\co(k+1),0]}(E(V)))\right|>0.\label{eq: lem - monotonicity of t_=00007Bc.m.n=00007D-1-1}
\end{equation}
\item \label{enu: lem-monotonicity of t_=00007Bc.m.n=00007D - sign of product}
If $E(V)\in B^{1}_{1}(V)$, then for all $n\in\N$, 
\[
\sgn\left(t_{\co(k+1)}(E(V))t_{[\co(k+1),n-1]}(E(V))t_{[\co(k+1),n]}(E(V))\right)=+1.
\]
\end{enumerate}
\end{lem}

\begin{proof}
Let $V>0$ be such that $E(V)\in I_{k}(V)$. For brevity, we remove
the $V$ dependence from the notations unless we want to emphasize
its dependence. For convenience of the reader, a sketch of the involved
spectral bands is provided in Figure~\ref{Fig: sign trace and central points}.
The order in which we prove the sections of the lemma is (\ref{enu: lem-monotonicity of t_=00007Bc.m.n=00007D - sign t_c}),
(\ref{enu: lem-monotonicity of t_=00007Bc.m.n=00007D - monotonicity t_=00005Bc.m.n=00005D}),
(\ref{enu: lem-monotonicity of t_=00007Bc.m.n=00007D - sign t_=00005Bc.m.n=00005D})
and (\ref{enu: lem-monotonicity of t_=00007Bc.m.n=00007D - sign of product}).
\begin{figure}[hbt]
\includegraphics[scale=1.15]{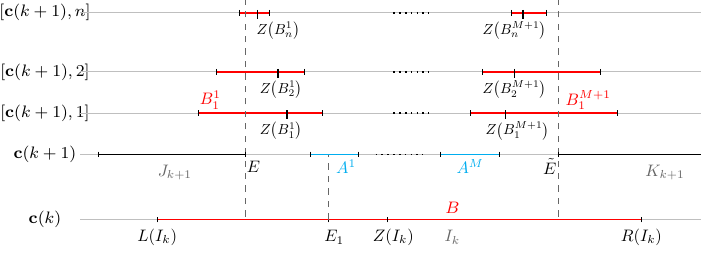}\caption{A sketch of the situation in Lemma~\ref{lem: monotonicity of t_=00007Bc.m.n=00007D}.\label{Fig: sign trace and central points}}
\end{figure}

(\ref{enu: lem-monotonicity of t_=00007Bc.m.n=00007D - sign t_c})
Since $I_{k}$ is of type $B$, it is proven in \cite[Proof of Lem.~3.3]{Raym95}
(see also \cite[Cor.~4.15]{BaBeBiTh22}), that there exists a unique
value $E_{1}\in I_{k}$ such that 
\begin{equation}
t_{\co(k)}(E_{1}):=2\cos\left(\frac{\pi}{c_{k+1}+1}\right)\sgn(t_{\co(k)}(L(I_{k})))\label{eq: proof of monotonicity - trace value in I_cm-1-1}
\end{equation}
and $E_{1}\in A^{1}$, where $A^{1}$ is the left-most spectral band
of type $A$ which is contained in $I_{k}$, see Proposition~\ref{prop: A-B-types-from-algebra-paper}\,(\ref{enu: forward-A}).
Note that it is essential to use here that $I_{k}$ is of type $B$
for the existence of such an $\tA$ band. Otherwise, if $I_{k}$ were
of type $A$ and $c_{k+1}=1$, then $I_{k}$ would not contain such
a spectral band $A^{1}$ and the argument would have failed; see counter-example
in Example~\ref{exa: CounterExample - passing middle point}.

Since $A^{1}\strict I_{k}$ and $J_{k+1}\prec I_{k}$ we get that
$J_{k+1}(V)\prec A^{1}(V)$ and $J_{k+1}(V)\cap A^{1}(V)=\emptyset$
since both spectral bands are in $\sigma_{\co(k+1)}$. This, together
with the notation $E:=R(J_{k+1})$, implies $E<E_{1}$. Recalling
that $\left.t_{\co(k)}\right|_{I_{k}}$ is strictly monotone and that
$t_{\co(k)}(\Cen(I_{k}))=0$, we get from (\ref{eq: proof of monotonicity - trace value in I_cm-1-1})
that if $c_{k+1}=1$, then $E_{1}=\Cen(I_{k})$ and if $c_{k+1}>1$,
then $E_{1}<\Cen(I_{k})$. In either of these cases, we conclude $E<\Cen(I_{k})$.
This estimate immediately implies $\sgn\left(t_{\co(k)}(E(V))\right)=\sgn\left(t_{\co(k)}\left(L(I_{k}(V))\right)\right)$
since $\left.t_{\co(k)}\right|_{I_{k}}$ is strictly monotone and
continuous and $E(V)\in I_{k}(V)$.

(\ref{enu: lem-monotonicity of t_=00007Bc.m.n=00007D - monotonicity t_=00005Bc.m.n=00005D})
The statement is proven by induction over $n\in\Nz$ simultaneously
for all $V>0$ satisfying $E(V)\in B^{1}_{1}(V)$. By (\ref{enu: lem-monotonicity of t_=00007Bc.m.n=00007D - sign t_c}),
we have $\left|t_{\co(k)}(E(V)))\right|>0$. Since $t_{[\co(k+1),0]}=t_{\co(k)}$
(see Lemma~\ref{lem: Equal-co-values}), the induction base is $\left|t_{[\co(k+1),1]}(E(V)))\right|>\left|t_{\co(k)}(E(V)))\right|$.
To show it we simplify notation, defining
\[
x(V):=t_{\co(k)}(E(V)),\quad y(V):=t_{\co(k+1)}(E(V))\quad\textrm{and}\quad z(V):=t_{[\co(k+1),1]}(E(V)),
\]
so that we aim to prove $\left|z(V)\right|>\left|x(V)\right|$. By
Proposition~\ref{prop: traceMaps}~(\ref{enu:.Prop-traceMaps-3 intersection property}),
we have $\sigma_{\co(k)}(V)\cap\sigma_{\co(k+1)}(V)\cap\sigma_{[\co(k+1),1]}(V)=\emptyset$
for $V>4$ implying in particular that $J_{k+1}(V)\cap I_{k}(V)=\emptyset$
and therefore $J_{k+1}(V)\prec I_{k}(V)$ for $V>4$. Using the continuity
of the spectral band edges (see e.g., (\ref{eq: Lipschitz_spectral_bands})),
we get that there is a $V_{1}>0$ such that $E(V_{1})=L(I_{k}(V_{1}))$
and for $V>V_{1}$, we have $E(V)<L(I_{k}(V))$ (and in particular
$E(V)\notin B^{1}_{1}(V)$). Since all spectral bands are either of
backward type $A$ or $B$ (strict inclusions), we conclude $E(V_{1})\not\in\sigma_{[\co(k+1),1]}(V_{1})$
implying
\begin{equation}
|z(V_{1})|=\left|t_{[\co(k+1),1]}(E(V_{1})))\right|>2=\left|t_{\co(k)}(E(V_{1}))\right|=|x(V_{1})|.\label{eq: proof of monotonicity - V_1 trace estimate-1}
\end{equation}

Recall that we aim to prove $|z(V)|>|x(V)|$ if $E(V)\in B^{1}_{1}(V)$
for $V>0$. Assume by contradiction that this is not true. Thus, by
continuity of these maps in $V$ and by (\ref{eq: proof of monotonicity - V_1 trace estimate-1}),
there is a $V_{0}>0$ satisfying 
\begin{equation}
E(V_{0})\in B^{1}_{1}(V_{0}),\quad|z(V_{0})|=|x(V_{0})|\quad\textrm{and}\quad|z(V)|>|x(V)|\quad\textrm{for }V>V_{0}.\label{eq: proof of monotonicity - V_0 assumptions-1}
\end{equation}
We will show that $V_{0}=0$ contradicting $V_{0}>0$. Since $E(V_{0})\in B^{1}_{1}(V_{0})\strict I_{k}(V_{0})$,
there is an $\varepsilon>0$ such that $E(V)\in I_{k}(V)$ for $V_{0}<V<V_{0}+\varepsilon$.
Since in this case $E(V)<\Cen(I_{k}(V))$ holds by (\ref{enu: lem-monotonicity of t_=00007Bc.m.n=00007D - sign t_c}),
we conclude $|x(V)|>0$ for $V_{0}<V<V_{0}+\varepsilon$. By the choice
of $V_{0}$, we have $|z(V)|>|x(V)|>0$ for $V_{0}<V<V_{0}+\varepsilon$
implying
\[
\sgn(z(V))=\sgn\left(t_{[\co(k+1),1]}\left(L(B^{1}_{1}(V))\right)\right).
\]
Hence, (\ref{enu: lem-monotonicity of t_=00007Bc.m.n=00007D - sign t_c})
together with Lemma~\ref{lem: Product traces - consecutive B-bands}
imply $\sgn\left(x(V)y(V)z(V)\right)=+1$. Moreover, $|y(V)|=\left|t_{\co(k+1)}(E(V))\right|=2$
follows from $E(V)=R\left(J_{k+1}(V)\right)$, by Proposition~\ref{prop: traceMaps}~(\ref{enu:.Prop-traceMaps-band edges}).
Combining these observations with the Fricke--Vogt invariant (see
Proposition~\ref{prop: traceMaps}~(\ref{enu:.Prop-traceMaps-Fricke-Vogt})),
we conclude
\begin{align*}
4+V^{2}=x(V)^{2}+y(V)^{2}+z(V)^{2}-x(V)y(V)z(V) & =x(V)^{2}+4+z(V)^{2}-2|x(V)z(V)|\\
 & =4+\left(|x(V)|-|z(V)|\right)^{2},
\end{align*}
for $V_{0}<V<V_{0}+\varepsilon$. If $V\searrow V_{0}$, then (\ref{eq: proof of monotonicity - V_0 assumptions-1})
leads to
\[
|V_{0}|^{2}=\lim_{V\searrow V_{0}}|V|^{2}=\lim_{V\searrow V_{0}}\left(|x(V)|-|z(V)|\right)^{2}=0,
\]
contradicting $V_{0}>0$. This proves the induction base.

For the induction step, suppose that $E:=R(J_{k+1})\in B^{1}_{1}$
and $\left|t_{[\co(k+1),n]}(E)\right|>\left|t_{[\co(k+1),n-1]}(E)\right|$
holds for some $n\geq1$. Using the recursive trace relation in Proposition~\ref{prop: traceMaps}~(\ref{enu:.Prop-traceMaps-Recursive Relation}),
we conclude
\begin{align*}
\left|t_{[\co(k+1),n+1]}(E)\right| & =\left|t_{\co(k+1)}(E)t_{[\co(k+1),n]}(E)-t_{[\co(k+1),n-1]}(E)\right|\\
 & \geq2\left|t_{[\co(k+1),n]}(E)\right|-\left|t_{[\co(k+1),n-1]}(E)\right|>\left|t_{[\co(k+1),n]}(E)\right|,
\end{align*}
where we used that $\left|t_{\co(k+1)}(E)\right|=2$ (since $E=R(J_{k+1})$)
and the induction assumption.

(\ref{enu: lem-monotonicity of t_=00007Bc.m.n=00007D - sign t_=00005Bc.m.n=00005D})
Like in (\ref{enu: lem-monotonicity of t_=00007Bc.m.n=00007D - sign t_c}),
it suffices to prove $E(V)<\Cen\left(B^{1}_{n}(V)\right)$. The statement
of the signs of the traces follows then directly. By \cite[Prop.~3.1 (iii)]{Raym95}
(see also \cite[Prop.~4.7]{BaBeBiTh22}), we have $\sigma_{\co(k)}(V)\cap\sigma_{\co(k+1)}(V)\cap\sigma_{[\co(k+1),1]}(V)=\emptyset$
for $V>4$ implying $J_{k+1}(V)\prec B^{1}_{1}(V)$ and $J_{k+1}(V)\cap B^{1}_{1}(V)=\emptyset$
for $V>4$. By the tower property $B^{1}_{n+1}\strict B^{1}_{n}$
(Proposition~\ref{prop: Tower Property}), we conclude
\[
E(V)<L\left(B^{1}_{1}(V)\right)<L\left(B^{1}_{n}(V)\right)<\Cen\left(B^{1}_{n}(V)\right),\qquad V>4.
\]
It is left to show that $E(V)<\Cen\left(B^{1}_{n}(V)\right)$ holds
for all $V>0$. By continuity we need to show that there exist no
$V'$ such that $E(V')=\Cen\left(B^{1}_{n}(V')\right)$. Assuming
by contradiction that there exists such $V'>0$ we have $E(V')\in B^{1}_{n}(V')$
and then $E(V')\in B^{1}_{1}(V')$ follows as $B^{1}_{n}\strict B^{1}_{1}$
holds by the tower property (Proposition~\ref{prop: Tower Property}).
By (\ref{enu: lem-monotonicity of t_=00007Bc.m.n=00007D - monotonicity t_=00005Bc.m.n=00005D})
and $E(V')\in B^{1}_{1}(V')$ we conclude $t_{[\co(k+1),n]}(E(V'))>0$
which contradicts $E(V')=\Cen\left(B^{1}_{n}(V')\right)$.

(\ref{enu: lem-monotonicity of t_=00007Bc.m.n=00007D - sign of product})
The case $n=1$ follows from (\ref{enu: lem-monotonicity of t_=00007Bc.m.n=00007D - sign t_c}),
(\ref{enu: lem-monotonicity of t_=00007Bc.m.n=00007D - sign t_=00005Bc.m.n=00005D})
and Lemma~\ref{lem: Product traces - consecutive B-bands}. The case
$n>1$ follows from (\ref{enu: lem-monotonicity of t_=00007Bc.m.n=00007D - sign t_=00005Bc.m.n=00005D})
and Lemma~\ref{lem: Product traces - along the tower}.
\end{proof}

We point out that the statement in Lemma~\ref{lem: monotonicity of t_=00007Bc.m.n=00007D}~(\ref{enu: lem-monotonicity of t_=00007Bc.m.n=00007D - sign t_c})
fails if $I_{k}$ is not of type $B$ and henceforth also the consecutive
statements are not necessarily true anymore, see Example~\ref{exa: CounterExample - passing middle point}
and the explanation in the beginning of the proof. Nevertheless, if
$c_{k+1}>1$, the statement extends verbatim to the case where $I_{k}$
is of type $A$, as is also explained within the proof (but we do
not need to apply this in the present work).

\section{Proof of the injectivity of $\protect\emap$ (in Theorem~\ref{thm: paths and spectra}
(\ref{enu: thm-paths and spectra - bijection}))\label{sec: Injectivity-Proof}}

In this section, we prove that the map $\emap(~\cdot~;V):\partial\tree\rightarrow\sigma(\Ham)$
is injective.
\begin{lem}
[trace estimates] \label{lem: estimating t_=00007Bc.m.n=00007D}
Let $\alpha\in[0,1]\setminus\Q$ with infinite continued fraction
expansion $(0,c_{1},c_{2},\ldots)$ and $k\in\N$. Consider a spectral
band $I$ in $\sigk$ of type $B$ with the unique associated spectral
band $B^{1}$ in $\sigma_{k+2}$ of type $B$ as in Proposition~\ref{prop: A-B-types-from-algebra-paper}.
Furthermore, let $J$ in $\sigkp$ be the rightmost spectral band
of $\sigkp$ for which $J\prec I$. 
\begin{enumerate}
\item \label{enu: lem-estimating t_=00007Bc.m.n=00007D-1}If $E(V):=R(J(V))\in B^{1}(V)$
for $V>0$, then 
\[
\left|t_{\co(k+2)}(E(V))\right|-\left|t_{\co(k)}(E(V))\right|=c_{k+2}V.
\]
\item \label{enu: lem-estimating t_=00007Bc.m.n=00007D-2}Let $\widetilde{J}$
be the spectral band in $\sigma_{k+3}$ such that $\widetilde{J}\strict J$,
and $\widetilde{J}$ is the rightmost band of $\sigma_{k+3}$ for
which $\widetilde{J}\prec B^{1}$. If $E(V):=R(J(V))\in B^{1}(V)$
and $\widetilde{E}(V):=R(\widetilde{J}(V))\in B^{1}(V)$ for $V>0$,
then 
\[
\left|t_{\co(k+2)}(\widetilde{E}(V))\right|>\left|t_{\co(k+2)}(E(V))\right|.
\]
\end{enumerate}
\end{lem}

\begin{proof}
By the tower property (Proposition~\ref{prop: Tower Property}),
$B^{1}$ may be written as $B^{1}_{n}$ with $n=c_{k+2}$ and we have
$E(V)\in B^{1}(V)\strict B^{1}_{1}(V)$ and so Lemma~\ref{lem: monotonicity of t_=00007Bc.m.n=00007D}~(\ref{enu: lem-monotonicity of t_=00007Bc.m.n=00007D - sign of product})
asserts that the product of the following traces is positive. Combining
this with the Fricke--Vogt invariant (Proposition~\ref{prop: traceMaps}~(b))
and $\left|t_{\co(k+1)}(E(V))\right|=2$, we conclude for $n\in\N$,
\begin{align*}
4+V^{2}= & \left(t_{\co(k+1)}(E(V))\right)^{2}+\left(t_{[\co(k+1),n-1]}(E(V))\right)^{2}+\left(t_{[\co(k+1),n]}(E(V))\right)^{2}\\
 & \quad\quad-t_{\co(k+1)}(E(V))t_{[\co(k+1),n-1]}(E(V))t_{[\co(k+1),n]}(E(V))\\
= & 4+\left(t_{[\co(k+1),n-1]}(E(V))\right)^{2}+\left(t_{[\co(k+1),n]}(E(V))\right)^{2}\\
 & \quad\quad-2\left|t_{[\co(k+1),n-1]}(E(V))t_{[\co(k+1),n]}(E(V))\right|\\
= & 4+\left(\left|t_{[\co(k+1),n]}(E(V))\right|-\left|t_{[\co(k+1),n-1]}(E(V))\right|\right)^{2}.
\end{align*}
Hence, Lemma~\ref{lem: monotonicity of t_=00007Bc.m.n=00007D}~(\ref{enu: lem-monotonicity of t_=00007Bc.m.n=00007D - monotonicity t_=00005Bc.m.n=00005D})
implies
\[
\left|t_{[\co(k+1),n]}(E(V))\right|-\left|t_{[\co(k+1),n-1]}(E(V))\right|=V.
\]
Summing the above for $n$ ranging from $1$ to $c_{k+2}$, then a
telescoping sum argument and $t_{[\co(k+1),0]}=t_{\co(k)}$ (Proposition~\ref{prop: traceMaps})
finishes the proof of the first part of the lemma.

To prove the second part of the lemma, let $\widetilde{E}(V):=R(\widetilde{J}(V))$
and note that $\widetilde{J}(V)\strict J_{k+1}(V)$ implies $\widetilde{E}(V)<E(V)$
for all $V>0$. Recalling that $\Cen(I)$ is the Zentrum of the interval
$I\subseteq\sigma_{k}$ (see Subsection~\ref{subsec: trace-monotonicity}),
Lemma~\ref{lem: monotonicity of t_=00007Bc.m.n=00007D}~(\ref{enu: lem-monotonicity of t_=00007Bc.m.n=00007D - sign t_c})
implies $E(V)<\Cen(B^{1}(V))$ and so $\widetilde{E}(V)<E(V)<\Cen(B^{1}(V))$.
Therefore we conclude that both $E(V)$ and $\widetilde{E}(V)$ are
contained in the ``left'' part of $B^{1}$, i.e., $E(V),\widetilde{E}(V)\in\left[L(B^{1}(V)),Z(B^{1}(V))\right]$.
This, together with $\widetilde{E}(V)<E(V)$ and with the strict monotonicity
of the trace $t_{\co(k+2)}$ on the spectral band $B^{1}$ yields
$\left|t_{\co(k+2)}(\widetilde{E}(V))\right|>\left|t_{\co(k+2)}(E(V))\right|.$
\end{proof}

\begin{rem}
\label{rem: estimating t_=00007Bc.m.n=00007D-1}One can prove the
symmetric cases of Lemma~\ref{lem: monotonicity of t_=00007Bc.m.n=00007D}
and Lemma~\ref{lem: estimating t_=00007Bc.m.n=00007D}: Let $I$
be a spectral band in $\sigma_{k}$ of type $\tB$ with the unique
associated spectral band $B^{M+1}_{1}$ as in Proposition~\ref{prop: Tower Property}.
Furthermore, let $K_{k+1}(V)$ in $\sigma_{\co(k+1)}(V)$ be the leftmost
spectral band of $\sigma_{\co(k+1)}(V)$ for which $I(V)\prec K_{k+1}(V)$
see an illustration in Figure~\ref{Fig: sign trace and central points}.

Set $E(V)=L(K_{k+1}(V))$. Following Lemma~\ref{lem: monotonicity of t_=00007Bc.m.n=00007D},
one can prove for all $n\in\N$,
\[
E(V)\in B^{M+1}_{1}(V)\quad\Rightarrow\quad\sgn\left(t_{\co(k+1)}(E(V))t_{[\co(k+1),n-1]}(E(V))t_{[\co(k+1),n]}(E(V))\right)=+1,
\]
and
\[
\left|t_{[\co(k+1),n]}(E(V)))\right|>\left|t_{[\co(k+1),n-1]}(E(V)))\right|>\ldots>\left|t_{[\co(k+1),1]}(E(V)))\right|>0.
\]

Let $\widetilde{K}(V)$ be the spectral band in $\sigma_{\co(k+3)}(V)$
such that $\widetilde{K}(V)\strict K_{k+1}(V)$, and $\widetilde{K}(V)$
is the leftmost band of $\sigma_{\co(k+3)}(V)$ for which $B^{M+1}_{1}(V)\prec\widetilde{K}(V)$.
With this, straightforward modifications of the proof of Lemma~\ref{lem: estimating t_=00007Bc.m.n=00007D}
lead to
\[
\left|t_{\co(k+2)}(E(V))\right|-\left|t_{\co(k)}(E(V))\right|=c_{k+2}V
\]
and 
\[
\widetilde{E}(V):=L(\widetilde{K}(V))\in B^{M+1}_{1}(V)\quad\Rightarrow\quad\left|t_{\co(k+2)}(\widetilde{E}(V))\right|>\left|t_{\co(k+2)}(E(V))\right|.
\]
These results can be used to improve the bound (\ref{eq: E_a injective - estimate coupling})
on the coupling constant $V$ which is used in the proof of Lemma~\ref{lem: map boundary-spectrum is injective}
(and may be of independent interest), see also Remark~\ref{rem: Estimate Coupling - B-band overlap}.
\end{rem}

\begin{example}
\label{exa: CounterExample - passing middle point} A spectral band
edge may pass the Zentrum of an adjacent spectral band one level below,
if this band is not of type $B$. Let $\co=[0,0]$ and $I_{k}(V)=[-2,2]$
in $\sigc(V)$. Then $I_{k}(V)$ is of type $A$ and the spectral
band $K(V)=[-2+V,2+V]$ in $\sigma_{[0,0,1]}(V)$ is of type $\tB$.
Since $t_{\co}(E,V)=E$, the Zentrum satisfies $\Cen(I_{k}(V))=0$.
Denoting $\tilde{E}(V):=L(K(V))=-2+V$, we have that $\tilde{E}(V)>\Cen(I_{k}(V))$
if $V>2$ and $\tilde{E}(V)<\Cen(I_{k}(V))$ if $V<2$. In particular,
$\sgn(R(I_{k}(V))\neq\sgn(\tilde{E}(V))$ whenever $V<2$. 
\end{example}

We finally prove the injectivity of the map $E_{\alpha}$.

\begin{lem}
\label{lem: map boundary-spectrum is injective} [the injectivity part of Theorem \ref{thm: paths and spectra}~(\ref{enu: thm-paths and spectra - bijection})]
Let $\alpha\in[0,1]\setminus\Q$ and $V>0$. Then the map $\emap(~\cdot~;V):\partial\tree\rightarrow\sigma(\Ham)$
is injective.
\end{lem}

\begin{proof}
Let $\gamma_{L}=(u^{L}_{0},u^{L}_{1},\ldots),\gamma_{R}=(w^{R}_{0},w^{R}_{1},\ldots)\in\partial\tree$
be different and without loss of generality assume $\gamma_{L}\preceq\gamma_{R}$.
We show that $\emap(\gamma_{L}~;V)\neq\emap(\gamma_{R}~;V)$. Note
that $u^{L}_{0}=w^{R}_{0}$ is the root of $\tree$. Since $\gamma_{1}\preceq\gamma_{2}$,
there is a $k_{0}\in\N_{0}$ such that $u^{L}_{j}=w^{R}_{j}$ for
$1\leq j\leq k_{0}$ and $u^{L}_{j}\prec w^{R}_{j}$ for $j>k_{0}$.

The proof is carried out in two steps. In step 1, we describe two
auxiliary paths $\gamma_{u}$ and $\gamma_{w}$ such that $\gamma_{L}\preceq\gamma_{u}\preceq\gamma_{w}\preceq\gamma_{R}$
and $\gamma_{u}\neq\gamma_{w}$. Hence, $E_{\alpha}(\gamma_{L};V)\leq E_{\alpha}(\gamma_{u};V)\leq E_{\alpha}(\gamma_{w};V)\leq E_{\alpha}(\gamma_{R};V)$
follows from Lemma~\ref{lem: order preserving E_=00005Calpha}. In
Step 2 we show that in fact $E_{\alpha}(\gamma_{u};V)\neq E_{\alpha}(\gamma_{w};V)$
finishing the proof.
\begin{figure}[hbt]
\includegraphics[scale=1.25]{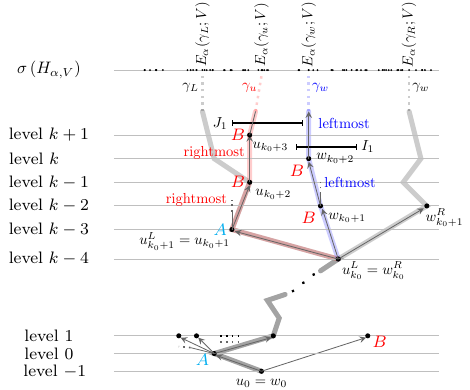}\caption{A sketch of the construction of the paths $\gamma_{u}$ and $\gamma_{w}$.}
\label{Fig: neighboring paths}
\end{figure}

\uline{Step 1:} The auxiliary paths $\gamma_{u}$ and $\gamma_{w}$
are recursively constructed, as described below (the main ideas are
sketched in Figure~\ref{Fig: neighboring paths}). Since $u^{L}_{k_{0}}=w^{R}_{k_{0}}$
and $u^{L}_{k_{0}+1}\prec w^{R}_{k_{0}+1}$, there is a unique vertex
$w_{k_{0}+1}$ satisfying
\begin{enumerate}[label=(\arabic*)]
\item  \label{enu: Proof of injective - choice first vertex gamma_L and gamma_R - 1}$u^{L}_{k_{0}+1}\prec w_{k_{0}+1}$
and there is an edge $u^{L}_{k_{0}}\to w_{k_{0}+1}$,
\item \label{enu: Proof of injective - choice first vertex gamma_L and gamma_R - 2}every
vertex $w$ satisfying \ref{enu: Proof of injective - choice first vertex gamma_L and gamma_R - 1}
fulfills either $w=w_{k_{0}+1}$ or $w_{k_{0}+1}\prec w$.
\end{enumerate}
Note that such a vertex exists since $w^{R}_{k_{0}+1}$ satisfies
\ref{enu: Proof of injective - choice first vertex gamma_L and gamma_R - 1}.
Define $u_{k_{0}+1}:=u^{L}_{k_{0}+1}$. By construction and the interlacing
property (Figure~\ref{Fig: TreeData}~(2)), we have
\begin{equation}
u_{k_{0}+1}\prec w_{k_{0}+1}\quad\textrm{and the vertices have different labels}.\label{eq: E_a injective - Property chosen vertex}
\end{equation}

Continue defining $\gamma_{u}$ as follows. For $j\in\N$, choose
$u_{k_{0}+j+1}$ to be the unique rightmost vertex such that there
is an edge $u_{k_{0}+j}\to u_{k_{0}+j+1}$, i.e. for any other vertex
$u$, for which there is an edge $u_{k_{0}+j}\to u$, we have $u\prec u_{k_{0}+j+1}$.
Then, $\gamma_{u}$ is defined by the path $(u^{L}_{0},\ldots,u^{L}_{k_{0}},u_{k_{0}+1},u_{k_{0}+2},\ldots)$.

Similarly, we define $\gamma_{w}$ but instead of choosing the rightmost
vertex, the leftmost is chosen. For $j\in\N$, $w_{k_{0}+j+1}$ is
the unique leftmost vertex such that there is an edge $w_{k_{0}+j}\to w_{k_{0}+j+1}$,
i.e. for any other vertex $w$, which admits an edge $w_{k_{0}+j}\to w$,
we have $w_{k_{0}+j+1}\prec w$. Then, $\gamma_{w}$ is defined by
the path $(w^{R}_{0},\ldots,w^{R}_{k_{0}},w_{k_{0}+1},w_{k_{0}+2},\ldots)$.

By construction, we get (as justified below) that for all $j\in\N$,
\begin{enumerate}
\item \label{enu: Proof of injective - prop 1}the vertices $u_{k_{0}+j+1}$
and $w_{k_{0}+j+1}$ are of type $B$,
\item \label{enu: Proof of injective - prop 2}the vertices $u_{k_{0}+j+1}$
and $w_{k_{0}+j+1}$ are in different but consecutive levels,
\item \label{enu: Proof of injective - prop 3}$u_{k_{0}+j+2}\prec w_{k_{0}+j+1}$
and there exists no other vertex $u$ with $u_{k_{0}+j+2}\prec u\prec w_{k_{0}+j+1}$.
\end{enumerate}
Statement (\ref{enu: Proof of injective - prop 1}) follows as the
leftmost (resp. rightmost) vertex $u$ connected to some vertex $w$
(except the root) is always labeled $B$ by definition of the branching,
see Figure~\ref{Fig: TreeData}~(2). By Equation~(\ref{eq: E_a injective - Property chosen vertex}),
the vertices $u_{k_{0}+1}$ and $w_{k_{0}+1}$ are in consecutive
levels. By (\ref{enu: Proof of injective - prop 1}), $u_{k_{0}+j+1}$
is two levels higher than $u_{k_{0}+j}$ and $w_{k_{0}+j+1}$ is two
levels higher than $w_{k_{0}+j}$. Thus, (\ref{enu: Proof of injective - prop 2})
follows inductively from (\ref{enu: Proof of injective - prop 1}).
Finally, (\ref{enu: Proof of injective - prop 3}) follows from construction
and the definition of the order.

From (\ref{eq: E_a injective - Property chosen vertex}) we get $\gamma_{u}\preceq\gamma_{w}$
and $\gamma_{u}\neq\gamma_{w}$. Since we choose for $\gamma_{u}$
the rightmost vertices (and for $\gamma_{w}$ the leftmost vertices)
in the construction, we have $\gamma_{L}\preceq\gamma_{u}$ and $\gamma_{w}\preceq\gamma_{R}$
(note that they can be equal). Thus, Lemma~\ref{lem: order preserving E_=00005Calpha}
implies
\[
E_{\alpha}(\gamma_{L};V)\leq E_{\alpha}(\gamma_{u};V)\leq E_{\alpha}(\gamma_{w};V)\leq E_{\alpha}(\gamma_{R};V).
\]

\uline{Step 2:} Let $V>0.$ We show $E_{\alpha}(\gamma_{u};V)\neq E_{\alpha}(\gamma_{w};V)$.
By definition of the map $E_{\alpha}$, it suffices to prove 
\begin{equation}
\left\{ E_{\alpha}(\gamma_{u};V)\right\} =\bigcap_{j\in\N_{0}}\left(\Psi(u_{k_{0}+1+j})\right)(V)\neq\bigcap_{j\in\N_{0}}\left(\Psi(w_{k_{0}+1+j})\right)(V)=\left\{ E_{\alpha}(\gamma_{w};V)\right\} .\label{eq: E_a injective - Intersection spectral bands - infinitely often}
\end{equation}
Let $(0,c_{1},c_{2},\ldots)$ be the continuous fraction expansion
of $\alpha$. Let $\ell\in\N$ be the level of the vertex $w_{k_{0}+2}$,
i.e. $\Psi\left(w_{k_{0}+2}\right)$ is a spectral band of $\sigma_{\ell}$.
We will prove that if for $m\in\N$, 
\begin{equation}
\left(\bigcap^{m}_{j=1}\left(\Psi(u_{k_{0}+j+1})\right)(V)\right)\bigcap\left(\bigcap^{m+1}_{j=1}\left(\Psi(w_{k_{0}+j+1})\right)(V)\right)\neq\emptyset\label{eq: E_a injective - Intersection spectral bands - m times}
\end{equation}
then
\begin{equation}
m\leq\sum^{m}_{j=1}c_{\ell+2j}<\frac{2}{V}.\label{eq: E_a injective - estimate coupling}
\end{equation}
Note that the first inequality is trivial since $c_{i}\geq1$ for
all $i\in\N$. After proving this statement, (\ref{eq: E_a injective - Intersection spectral bands - infinitely often})
implies (\ref{eq: E_a injective - Intersection spectral bands - m times})
for all $m\in\N$. Thus, $V>0$ implies that (\ref{eq: E_a injective - Intersection spectral bands - infinitely often})
is false by contraposition.  This part of the proof is based on Lemma~\ref{lem: estimating t_=00007Bc.m.n=00007D}.

Suppose (\ref{eq: E_a injective - Intersection spectral bands - m times})
holds for $m\in\N$. In order to simplify the notation, denote for
$j\in\N_{0}$, 
\[
J_{j}(V):=\left(\Psi(u_{k_{0}+j+2})\right)(V)\quad\textrm{and}\quad I_{j}(V):=\left(\Psi(w_{k_{0}+j+1})\right)(V),
\]
see Figure~\ref{Fig: Interlacing B bands}.

When using the notation above, we assume without loss of generality
that $u_{k_{0}+j+2}$ is one level higher than $w_{k_{0}+j+1}$ (those
levels are $\ell+2j-1$ and $\ell+2j-2$, correspondingly). Such a
case is depicted in Figure~\ref{Fig: neighboring paths}. It might
also be that $u_{k_{0}+j+2}$ is one level lower than $w_{k_{0}+j+1}$
(this would be the case if $u_{k_{0}+1}$ is of type $\tB$ and $w_{k_{0}+1}$
is of type $\tA$). Then we would adapt the notation above by writing
$J_{j-1}(V):=\left(\Psi(u_{k_{0}+j+2})\right)(V)$ and with no change
in the notation of $I_{j}(V)$. In either of these two cases, we get
that (\ref{enu: Proof of injective - prop 1}), (\ref{enu: Proof of injective - prop 2})
and (\ref{enu: Proof of injective - prop 3}) in step $1$ inductively
imply (using Theorem~\ref{thm: vertices are spectral bands}) that
for $j\in\N$,
\begin{itemize}
\item $I_{j}$ is a spectral band of $\sigma_{\co(\ell+2j-2)}$ and $I_{j+1}\strict I_{j}$,
\item $J_{j}$ is a spectral band of $\sigma_{\co(\ell+2j-1)}$ satisfying
$J_{j+1}\strict J_{j}$,
\item $J_{j}$ is the rightmost band $\sigma_{\co(\ell+2j-2)}$ satisfying
$J_{j}\prec I_{j}$,
\item $I_{j+1}$ equals $B^{1}$ where $B^{1}$ is the spectral band in
$\sigma_{\co(\ell+2j)}$ of type $B$ associated with $I_{j}$ as
in Proposition~\ref{prop: A-B-types-from-algebra-paper}.
\end{itemize}
\begin{figure}[hbt]
\includegraphics[scale=0.95]{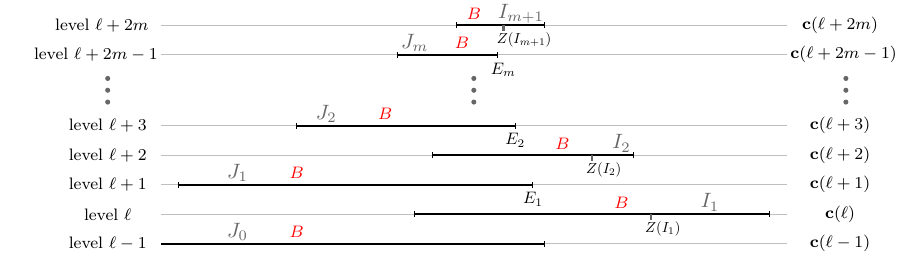}\caption{A sketch of the spectral bands $I_{1},\ldots,I_{m+1}$ and $J_{0},\ldots,J_{m}$
introduced in the proof of Lemma~\ref{lem: map boundary-spectrum is injective}.}
\label{Fig: Interlacing B bands}
\end{figure}

Set $E_{j}:=R(J_{j})$ to be the right edge of $J_{j}$. These properties
are sketched in Figure~\ref{Fig: Interlacing B bands} and allow
us to apply Lemma~\ref{lem: estimating t_=00007Bc.m.n=00007D}. Suppose
(\ref{eq: E_a injective - Intersection spectral bands - m times})
holds for $m\in\N$ and let $1\leq j\leq m$. In the notation of that
lemma, we set $J(V):=J_{j}(V)$, $\widetilde{J}(V):=J_{j+1}(V)$,
$I(V):=I_{j}(V)$ and $B^{1}(V):=I_{j+1}(V)$. Then (\ref{eq: E_a injective - Intersection spectral bands - m times})
implies $E_{j}(V)=R\left(J(V)\right)\in B^{1}(V)$ and $E_{j+1}(V)=R\left(\widetilde{J}(V)\right)\in B^{1}(V)$.
Thus, Lemma~\ref{lem: estimating t_=00007Bc.m.n=00007D} implies
\begin{equation}
c_{\ell+2j}V=\left|t_{\co(\ell+2j)}(E_{j}(V))\right|-\left|t_{\co(\ell+2j-2)}(E_{j}(V))\right|,\qquad1\leq j\leq m,\label{eq: E_a injective - difference trace =00003D c_k |V|}
\end{equation}

and
\begin{equation}
\left|t_{\co(\ell+2j)}(E_{j+1}(V))\right|>\left|t_{\co(\ell+2j)}(E_{j}(V))\right|\quad1\leq j\leq m-1.\label{eq: E_a injective - monotonicity trace E_j}
\end{equation}
Summing Equation~(\ref{eq: E_a injective - difference trace =00003D c_k |V|})
for all $1\leq j\leq m$ and reordering the summands leads to
\begin{align*}
\sum^{m}_{j=1}c_{\ell+2j}V & =\sum^{m}_{j=1}\left|t_{\co(\ell+2j)}(E_{j}(V))\right|-\left|t_{\co(\ell+2j-2)}(E_{j}(V))\right|\\
 & =\left|t_{\co(\ell+2m)}(E_{m}(V))\right|-\underbrace{\left|t_{\co(\ell)}(E_{1}(V))\right|}_{\geq0}+\sum^{m-1}_{j=1}\underbrace{\left|t_{\co(\ell+2j)}(E_{j}(V))\right|-\left|t_{\co(\ell+2j)}(E_{j+1}(V))\right|}_{<0\textrm{ by Equation~(\ref{eq: E_a injective - monotonicity trace E_j})}}\\
 & <\left|t_{\co(\ell+2m)}(E_{m}(V))\right|.
\end{align*}
Since Equation~(\ref{eq: E_a injective - Intersection spectral bands - m times})
holds for $m\in\N$, we conclude $E_{m}(V)\in I_{m+1}(V)\subseteq\sigma_{\ell+2m}(V)$
and use it in the inequality above to get 
\[
\sum^{m}_{j=1}c_{\ell+2j}V<\left|t_{\co(\ell+2m)}(E_{m}(V))\right|\leq2.
\]
This proves that (\ref{eq: E_a injective - Intersection spectral bands - m times})
implies (\ref{eq: E_a injective - estimate coupling}).
\end{proof}

\begin{rem}
\label{rem: map boundary-spectrum is injective}Note that if $V>4$,
then proving the injectivity of $\emap(\cdot;V)$ is substantially
shorter. Specifically, Proposition~\ref{prop: traceMaps}\ (\ref{enu:.Prop-traceMaps-3 intersection property})
asserts that if $V>4$ then $\sigc(V)\cap\sigcm(V)\cap\sigma_{[\co,m-1]}(V)=\emptyset$.
Together with the tower property (Proposition~\ref{prop: Tower Property}),
one can deduce that if $V>4$ and $u,w\in\tree$ are two vertices
which are not connected by a directed path then $\Psi(u)(V)\cap\Psi(w)(V)=\emptyset$
(see also \cite[Lem.~5.12]{BaBeBiTh22}). Now, (\ref{eq: E_a injective - Intersection spectral bands - infinitely often})
in the proof above follows immediately if $V>4$.
\end{rem}

\begin{rem}
\label{rem: Estimate Coupling - B-band overlap}We observe that the
upper bound in (\ref{eq: E_a injective - estimate coupling}) in the
proof of Lemma~\ref{lem: map boundary-spectrum is injective} may
be improved, using Remark~\ref{rem: estimating t_=00007Bc.m.n=00007D-1}.
Specifically, it can be shown that if (\ref{eq: E_a injective - Intersection spectral bands - m times})
holds for $m\in\N$ and $V>0$, then
\[
2m\leq\sum^{2m}_{j=1}c_{\ell+j}<\frac{2}{V}.
\]
As in the proof of Lemma~\ref{lem: map boundary-spectrum is injective},
$m$ dictates how many spectral bands of type $B$ overlap (at least
$2m+2$) and $\ell$ is the level of the vertex $w^{R}_{k_{0}+2}$.
For example, these bounds imply that that $4$ bands of type $B$
can overlap only if $V<1$; this bound on $V$ is even smaller if
the digits $c_{\ell+1},c_{\ell+2}>1$. The bounds we provide here
may be considered as a refinement of Proposition~\ref{prop: traceMaps}\ (\ref{enu:.Prop-traceMaps-3 intersection property})
about the empty intersection of three spectra if $V>4$, which turned
out very useful in earlier works. Further note that we only provided
here a rough estimate which is enough for our purpose to prove injectivity.
However, these estimates may be further refined. This might be useful
for obtaining estimates on the Hausdorff dimension of the spectrum.
\end{rem}

\begin{proof}[Proof of Theorem~\ref{thm: paths and spectra}]
 Theorem~\ref{thm: paths and spectra} is a combination of Lemma~\ref{lem: Lipschitz continuous limiting spectrum},
Lemma~\ref{lem: IDOS is independent of V}, Lemma~\ref{lem: map boundary-spectrum is surjective},
Lemma~\ref{lem: order preserving E_=00005Calpha} and Lemma~\ref{lem: map boundary-spectrum is injective}.
\end{proof}

\subsection*{Acknowledgments}

 We are grateful to David Damanik, Jake Fillman and Anton Gorodetski
for inspiring discussions on the developments of the DTMP for Sturmian
systems. We thank Barak Biber and Yannik Thomas for extensive discussions
on the work of Laurent Raymond helping us to get a deeper insights.
We are thankful to Michael Baake for organizing and hosting a joint
meeting with Laurent Raymond in Bielefeld on December 2023. We thank
Laurent Raymond for his support and our interesting discussions.

We thank the Israel Institute of Technology and the University of
Potsdam for providing excellent working conditions during our mutual
visits. The final part of this work was done during a joint visit
at the Simon center in Stony Brook on June 2022. This happened during
a Workshop on \emph{Ergodic Operators and Quantum Graphs} organized
by David Damanik, Jake Fillman and Selim Sukhtaiev. SB was partially
supported by the Deutsche Forschungsgemeinschaft {[}BE 6789/1-1 to
SB{]} and the Maria-Weber Grant 2022 offered by the Hans Böckler Stiftung.
RB was supported by the Israel Science Foundation (ISF Grants No.
844/19 and 2362/25).

\bibliographystyle{amsalpha}
\bibliography{references}

\appendix
\appendix
\renewcommand{\thesection}{\Roman{section}}

\section{Sturmian dynamical systems\label{App: Sturmian dynamical systems}}

This appendix contains a very short description of Sturmian dynamical
systems. A thorough background may be found in the books \cite{Fogg_book02,Loth02,DaFi24-book_2}.
The sequences
\[
\omega_{\alpha}(n):=\chi_{[1-\alpha,1[}(n\alpha\mod 1),\quad n\in\N,~\alpha\in[0,1],
\]
are called \emph{mechanical words} \cite[Sec.~2.1.2]{Loth02}. If
$\alpha\not\in\Q$, $\omega_{\alpha}$ is also called a \emph{Sturmian}\emph{
sequence}. They naturally define a dynamical system as follows. Let
$\Aa:=\{0,1\}$ be equipped with the discrete topology and $\Aa^{\Z}:=\{\omega:\Z\to\Aa\}$
be the compact metrizable space equipped with the product topology.
Consider the shift $T:\Aa^{\Z}\to\Aa^{\Z},~(T\omega)(n):=\omega(n-1),~n\in\Z,$
being a homeomorphism. This induces a continuous group action $\Z\curvearrowright\Aa^{\Z}$
via $(n,\omega)\mapsto T^{n}\omega$. For $\alpha\in[0,1]$, we have
$\omega_{\alpha}\in\{0,1\}^{\Z}$ and its associated orbit closure
(in the product topology)
\[
\Omega_{\alpha}:=\overline{\textrm{Orb}(\omega_{\alpha})}:=\overline{\set{T^{n}\omega_{\alpha}}{n\in\Z}}
\]
defines a dynamical system $\Z\curvearrowright\Omega_{\alpha}$. The
dynamical system $\Omega_{\alpha}$ is minimal (namely for all $\omega\in\Omega_{\alpha}$,
we have $\Omega_{\alpha}:=\overline{\textrm{Orb}(\omega)}$) and uniquely
ergodic (it admits a unique shift invariant probability measure).
Particular elements of $\Omega_{\alpha}$ are the sequences $\omega_{\alpha,\xi}\in\{0,1\}^{\Z}$
for $\xi\in[0,1]$ defined by $\omega_{\alpha,\xi}(n):=\chi_{[1-\alpha,1[}(\xi+n\alpha\mod 1)$,
which are also called \emph{mechanical words}. For $\alpha,\xi\in[0,1]$
and $V\in\R$, consider the self-adjoint operator $H_{\alpha,V,\xi}:\ell^{2}(\Z)\to\ell^{2}(\Z)$
defined by
\[
(H_{\alpha,V,\xi}\psi)(n):=\psi(n+1)+\psi(n-1)+V\omega_{\alpha,\xi}(n)\thinspace\psi(n).
\]
If $\xi=0$, this operator coincides with $H_{\alpha,V}$ defined
in Equation~(\ref{eq: Hamiltonian defined}). Let $\alpha\in[0,1]$
and $V\in\R$ be fixed. Since $\omega_{\alpha,\xi}\in\Omega_{\alpha}$
and $\Omega_{\alpha}$ is minimal \cite[Cor.~10.2.17]{DaFi24-book_2},
the spectrum $\sigma(H_{\alpha,V,\xi})$ is independent of $\xi\in[0,1]$
and coincides with $\sigma(H_{\alpha,V}).$ Therefore, we set $\xi=0$
throughout this work.

Moreover, since the dynamical system $\Omega_{\alpha}$ is uniquely
ergodic, the limit 
\[
\lim_{n\to\infty}\frac{\#\set{\lambda\in\sigma\left(H_{\alpha,V,\xi}|_{[0,n-1]}\right)}{\lambda\leq E}}{n}
\]

exists and is independent of $\xi\in[0,1]$, see \cite[Thm.~4.9.11]{DaFi22-book_1},
\cite[Cor.~10.2.17]{DaFi24-book_2}. Therefore the IDS (\ref{eq: definition of DOS})
exists and we can set $\xi=0$.

\section{An explicit formula of the IDS via the spectral tree \label{App: IDS-explicit-formula}}

In Theorem~\ref{thm: paths and spectra}~(\ref{enu: thm-paths and spectra - IDOS})
it is stated that there exists a function $N_{\alpha}:\partial\tree\rightarrow[0,1]$
such that for all $V>0$, $\IDS\left(\emap(\gamma;V)\right)=N_{\alpha}(\gamma)$.
This statement is proven in Lemma~\ref{lem: IDOS is independent of V}.
Furthermore, one may provide an explicit expression of this function.
This is done in \cite{Raym95} for $V>4$ using a coding scheme and
we shortly present here an adaptation of this expression using infinite
paths, i.e., elements of $\partial\tree$. Towards this, we define
for $\gamma=(u_{0},u_{1},\ldots)\in\partial\tree$, the \emph{level
function} $\ell:\N_{0}\to\N_{-1}$ by setting $\ell(j)$ to be the
level of $u_{j}$. Next define the \emph{relative $A$-index} $\relA(\gamma,j)$
to be the number of vertices $w$ in level $\ell(j)+1$ admitting
an edge $(u_{j},w)\in\mathcal{E}_{\alpha}$ (i.e. $w$ has label $A$)
such that $w\prec u_{j+1}$, see Figure~\ref{fig: IDS pi's}. Note
that by the tree construction
\[
\relA(\gamma,j)\in\begin{cases}
\left\{ 0,\ldots,c_{\ell(j)+1}-1\right\}  & \textrm{if }u_{j}\textrm{ has label }A,\\
\left\{ 0,\ldots,c_{\ell(j)+1}\right\}  & \textrm{if }u_{j}\textrm{ has label }B,
\end{cases}
\]
and $\relA(\gamma,j)=0$ if and only if $u_{j+1}$ is either the leftmost
vertex with label $B$ connected to $u_{j}$ or the leftmost vertex
with label $A$ connected to $u_{j}$. Finally, set

\[
\delta_{A}(\gamma,j):=\begin{cases}
1 & \textrm{if }u_{j}\textrm{ has label }A,\\
0 & \textrm{if }u_{j}\textrm{ has label }B.
\end{cases}
\]
\begin{figure}[tb]
\includegraphics[scale=0.9]{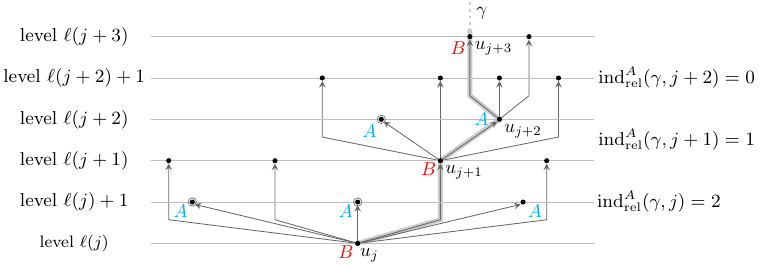} \caption{The figure demonstrating the notation introduced in the appendix.
The circled vertices $\odot$ are the ones that are counted for the
corresponding relative $A$-index. Note that not all vertices are
plotted in each level.}
\label{fig: IDS pi's}
\end{figure}
 With these notations at hand, the function $N_{\alpha}:\partial\tree\rightarrow[0,1]$
is written explicitly as 
\[
N_{\alpha}(\gamma)=-\alpha+\sum_{j\in\N_{0}}(-1)^{\ell(j)}\left(\relA(\gamma,j)+\delta_{A}(\gamma,j)\right)\thinspace\left(q_{\ell(j)}\alpha-p_{\ell(j)}\right),
\]
where $p_{k},q_{k}$ are coprime such that $\alpha_{k}=\frac{p_{k}}{q_{k}}$,
see Equation~(\ref{eq: finite continued fraction expansion}). This
equality is an immediate consequence of Lemma~\ref{lem: IDOS is independent of V}
and \cite[Thm. 4.7]{Raym95}. Note that this explicit representation
is a crucial ingredient in Raymond's work to prove that all gaps are
there for $V>4$. We refer the reader also a more detailed discussion
on that in \cite[Sec.~5.3, Prop.~5.21]{BaBeBiTh22}.
\end{document}